\documentclass{article}
\usepackage{amsmath,amssymb,amsfonts}%
\usepackage{amsthm}
\usepackage[margin=1in]{geometry}
\usepackage{graphicx} %
\usepackage{url}
\usepackage{mathrsfs}%
\usepackage{physics}
\usepackage[ruled,vlined]{algorithm2e}
\usepackage{bm}
\usepackage{multirow}

\newcommand{\shortauthors}[1]{}

\newcommand{\keywords}[1]{\par\noindent\textbf{Keywords:} #1\par}

\DeclareMathOperator{\UCRY}{UCRY}

\usepackage{tikz}
\usepackage{quantikz}
\definecolor{customgreen}{RGB}{0,150,0} 

\newtheorem{theorem}{Theorem}[section]

\newtheorem{corollary}[theorem]{Corollary}

\newtheorem{example}[theorem]{Example}
\newtheorem{lemma}[theorem]{Lemma}

\newtheorem{proposition}[theorem]{Proposition}

\newtheorem{remark}[theorem]{Remark}

\newcommand{\CC}{\mathsf{C}\mathsf{C}}

\newcommand{\rI}{\mathsf{I}}

\newcommand{\rR}{\mathsf{R}}

\newcommand{\rU}{\mathsf{U}}

\newcommand{\rA}{\mathsf{A}}

\newcommand{\rC}{\mathsf{C}}

\newcommand{\RY}{\mathrm{R}_y}

\newcommand{\CNOT}{\mathrm{CNOT}}

\usepackage{txfonts}

\numberwithin{equation}{section}

\makeatletter
\renewcommand{\@biblabel}[1]{#1\hfill \hspace{-0.2cm}}
\makeatother

\begin{document}

\title{A Rigorous and Self-Contained Proof of the Grover–Rudolph State Preparation Algorithm}

\author{%
Antonio Falc\'o\\
Departamento de Matem\'aticas, F\'{\i}sica y Ciencias Tecnol\'ogicas\\
Universidad Cardenal Herrera-CEU, CEU Universities\\
San Bartolom\'e 55, 46115 Alfara del Patriarca (Valencia), Spain\\
\texttt{afalco@uchceu.es}
\and
Daniela Falc\'o--Pomares\\
Grupo de Investigaci\'on Bisite, Universidad de Salamanca\\
Calle Espejo s/n, 37007 Salamanca (Spain)
\and
Hermann G. Matthies\\
Institute of Scientific Computing\\
Technische Universit\"at Braunschweig, \\ 
Universit\"atsplatz 2, 38106 Braunschweig, Germany
}

\maketitle

\begin{abstract}
We give a rigorous and self-contained analysis of the Grover--Rudolph quantum state-preparation
algorithm, which encodes a probability distribution $\{p_k\}$ as an $n$-qubit amplitude state
$\sum_k\sqrt{p_k}\ket{k}$ via a hierarchy of controlled $\RY$ rotations determined by a dyadic
refinement of the target.
We formalize the dyadic probability tree, derive the trigonometric factorization of conditional
masses, and prove by induction that the circuit prepares exactly the desired measurement law.
We further prove that perturbing each rotation angle by at most $\eta$ changes the output
distribution by at most $\min(1,n\eta)$ in total variation, and combine this with a Hoeffding
concentration bound to obtain an explicit design rule: $b\ge\log_2(2n\pi/\varepsilon)$ bits and
$S\ge 2^{n+1}\log(2/\delta)/\varepsilon^2$ shots suffice to achieve accuracy $\varepsilon$
with confidence $1-\delta$.
As a circuit-theoretic complement, we provide an ancilla-free transpilation of each stage into
$\{\RY(\cdot),X,\CNOT\}$ via Gray-code ladders and a Walsh--Hadamard angle transform.
\end{abstract}


\section{Introduction}

Preparing quantum states whose amplitudes encode classical data is a basic primitive in quantum
algorithm design. In particular, \emph{amplitude encoding} of probability distributions underlies
quantum Monte Carlo–type routines and amplitude-estimation methods, where expectation values or
integrals are converted into success probabilities and then estimated quadratically faster than
classically in idealized settings (see, e.g., the amplitude-estimation framework of
Brassard--H{\o}yer--Mosca--Tapp~\cite{Brassard2000}).

A particularly influential and widely reused state-preparation procedure is due to Grover and
Rudolph~\cite{GroverRudolph2002}. Given a probability distribution over $2^n$ outcomes (or, more generally, a nonnegative function on
$[0,1]$ sampled on a dyadic grid), their construction prepares an $n$-qubit state
$\sum_{z\in\mathbb{Z}_2^n}\sqrt{p_z}\,\ket{z}$ by recursively refining a dyadic partition and applying, at each
level, a family of controlled one-qubit rotations whose angles are defined from conditional masses on a
binary tree. The algorithm is conceptually elegant and has become a standard reference point for
state preparation from classical distributions.

Despite its popularity, the Grover--Rudolph procedure is often presented at a high level, with correctness
justified informally or with implicit assumptions about the dyadic refinement and the controlled-gate
semantics. The first objective of this paper is therefore \emph{mathematical}: we provide a rigorous and
self-contained treatment of the dyadic probability tree, the induced angle map, and the resulting
amplitude identities, culminating in a fully explicit proof that the circuit prepares the desired target
state. In particular, we isolate the key telescoping/trigonometric factorizations and prove them by
transparent induction, so that the overall correctness argument can be read independently of external
compilation folklore.

Grover and Rudolph's original note~\cite{GroverRudolph2002} introduced the recursive dyadic construction that we formalize here. Beyond this specific algorithm,
state-preparation and multiplexing problems have been extensively studied in the circuit-synthesis
literature. A prominent line of work concerns \emph{uniformly controlled one-qubit gates} (also called
\emph{multiplexed} or \emph{uniformly controlled rotations}), for which Gray-code techniques yield
systematic decompositions into elementary gates without ancillas. In particular, the Gray-code ladder
idea appears as a standard tool for implementing uniformly controlled rotations and related multiqubit
structures (see, e.g., the uniformly controlled gate framework in
Bergholm~\emph{et al.}~\cite{Bergholm2005} and the companion discussion of uniformly
controlled one-qubit gates in~\cite{Mottonen2004}).

Since the original Grover--Rudolph proposal, the state-preparation problem has continued to attract significant attention, with contributions addressing improved gate complexity for general and structured states.
Zhang, Li, and Yuan~\cite{Zhang2022} showed that any $n$-qubit state can be prepared with $\Theta(n)$-depth circuits at the cost of exponentially many ancilla qubits, and achieved $\Theta(\log(nd))$-depth for sparse states with $d$ nonzero amplitudes.
Marin-Sanchez, Gonzalez-Conde, and Sanz~\cite{MarinSanchez2023} developed quantum circuits that efficiently encode smooth functions through a dyadic-tree approach closely related to the Grover--Rudolph construction, making the connection between classical function approximation and amplitude encoding more explicit.
Gonzalez-Conde, Watts, Rodriguez-Grasa, and Sanz~\cite{GonzalezConde2024} extended this line by proposing efficient amplitude encoding for polynomial functions via block-encoding techniques.
Mao, Tian, and Sun~\cite{Mao2024} established a nearly optimal circuit size of $O(nd/\log(n+m)+n)$ for sparse state preparation using $m$ ancillary qubits, improving upon prior work by a factor of $\log d$.
More recently, Luo, Li, and Li~\cite{Luo2025} established an explicit depth-ancilla trade-off for sparse state preparation, achieving circuit depth $O\!\bigl(\frac{nd\log m}{m\log(m/n)}+\log nd\bigr)$ with $m\ge 6n$ ancilla qubits.
Ramacciotti, Lefterovici, and Rotundo~\cite{Ramacciotti2024} analyzed the Grover--Rudolph algorithm directly in the sparse regime, proving that the gate complexity is linear in the number of nonzero amplitudes and linear in the number of qubits after a simple modification.

The present paper differs from these works in three complementary ways.
First, our primary objective is \emph{mathematical rigour}: we provide a fully self-contained correctness
proof of the original Grover--Rudolph construction in the general dense, amplitude-encoded setting,
making explicit the dyadic refinement structure, the angle assignment, and the inductive amplitude
argument that is often left implicit in the literature.
Second, we prove a quantitative stability estimate for imperfectly implemented Grover--Rudolph angles,
showing how deterministic angle errors propagate to the final output distribution in total variation
distance.
Third, we integrate the circuit-synthesis perspective into the same rigorous framework: each
Grover--Rudolph stage is identified as a uniformly controlled \(\RY\) rotation, and an explicit
ancilla-free Gray-code transpilation is derived through a Walsh--Hadamard angle transform. Although
Gray-code decompositions are standard in quantum compilation, our contribution is to connect them
directly with the dyadic Grover--Rudolph hierarchy and to provide complete correctness proofs in this
specific setting.

The main contributions can be summarized as follows:
\begin{itemize}
\item A rigorous dyadic-tree formalization of the Grover--Rudolph angle construction, including explicit
      refinement identities and a clean inductive correctness proof of the prepared amplitudes.
\item A quantitative stability theorem for imperfectly implemented Grover--Rudolph angles, showing that
      if each angle is perturbed by at most \(\eta\), then the resulting output distribution changes by
      at most \(\min(1,n\eta)\) in total variation distance. This separates deterministic angle-synthesis
      errors from finite-shot statistical fluctuations.
\item A stage-wise circuit formulation showing that each Grover--Rudolph stage is a uniformly controlled
      one-qubit rotation, together with an explicit ancilla-free transpilation into
      \(\{\RY(\cdot),X,\CNOT\}\) using a Gray-code ladder and a Walsh--Hadamard angle transform, with
      full correctness proofs and implementable pseudo-code.
\item A systematic sensitivity study of the total-variation error as a function of angle-quantization
      precision \((b\in\{8,16,32\})\) and number of measurement shots \((S\in\{256,1024,4096\})\)
      for \(n\in\{2,3,4\}\) qubits, showing that shot noise dominates at practically relevant
      shot counts and that \(b=8\)-bit angle precision already reaches the shot-noise floor in the
      tested instances.
      These empirical findings are underpinned by a combined theoretical bound (Corollary~\ref{cor:combined_bound})
      that unifies the stability theorem with the Hoeffding inequality and yields an explicit
      design rule: $b\ge\log_2(2n\pi/\varepsilon)$ bits and $S\ge 2^{n+1}\log(2/\delta)/\varepsilon^2$
      shots suffice to guarantee $\mathrm{TV}(p,\hat p)\le\varepsilon$ with confidence $1-\delta$.
\end{itemize}

The paper is organized as follows.
Section~\ref{sec:prelim_main} fixes notation and collects the required preliminaries, including the
dyadic partition formalism for probability masses, the embedding conventions for one-qubit gates, and
the selective pattern-controlled gate construction used throughout; it also states the Grover--Rudolph
state-preparation circuit and the main theorem.
Section~\ref{sec:proof_GR} contains the rigorous and self-contained correctness proof, centered on the
trigonometric factorization of dyadic masses and the inductive amplitude formula that propagates across
Grover--Rudolph stages.
Section~\ref{sec:stability} proves the stability estimate for imperfect rotation angles and derives the
finite-precision consequence for angle synthesis, thereby separating deterministic synthesis error from
finite-shot sampling error.
In Section~\ref{sec:numerics} we provide a worked example together with an idealized noise-free
circuit-level verification of the construction, as well as a systematic sensitivity analysis of the
total-variation error as a function of angle-quantization precision and number of measurement shots.
Finally, Section~\ref{sec:transpilation_triangulum} develops the active-register viewpoint and gives an
explicit ancilla-free transpilation of each Grover--Rudolph stage into the elementary gate dictionary
\(\{\RY(\cdot),X,\CNOT(\cdot\to\cdot)\}\) via Gray-code ladder decompositions, including implementable
pseudo-code and circuit diagrams.

\section{Preliminary results and statement of the main theorem}\label{sec:prelim_main}

Throughout the paper, we write $\mathbb{Z}_2:=\{0,1\}$ and, for $n\ge 1$,
\[
\mathbb{Z}_{2^n}:=\{0,1,\dots,2^n-1\},\qquad N:=2^n.
\]
For a bitstring $\mathbf{z}=z_1\cdots z_m\in\mathbb{Z}_2^m$ we use the shorthand
$|\mathbf{z}\rangle:=|z_1\rangle\otimes\cdots\otimes|z_m\rangle$ for the corresponding
computational basis vector.

\subsection{Binary representations and dyadic partitions}\label{subsec:dyadic}

For $k\in\mathbb{Z}_{2^n}$, let $b_n(k)=z_1z_2\cdots z_n\in\mathbb{Z}_2^n$ denote the
(length-$n$) binary representation of $k$, i.e.
\[
k=b_n^{-1}(z_1z_2\cdots z_n) = z_1 2^0+z_2 2^1+\cdots+z_n 2^{n-1}.
\]
For each level $1\le \ell\le n$ and $0\le k\le 2^\ell$, define the dyadic grid points
\[
z^{(\ell)}_{b_\ell(k)}:=\frac{k}{2^\ell},
\qquad\text{with the convention } z^{(\ell)}_{b_\ell(2^\ell)}:=1.
\]
We then define the dyadic intervals
\[
I^{(\ell)}_{b_\ell(k)}:=\Big[z^{(\ell)}_{b_\ell(k)},\,z^{(\ell)}_{b_\ell(k+1)}\Big]
=\Big[\frac{k}{2^\ell},\,\frac{k+1}{2^\ell}\Big],
\qquad k\in\mathbb{Z}_{2^\ell}.
\]
At the finest level $\ell=n$, we write
\[
I_k:=I^{(n)}_{b_n(k)}=\Big[\frac{k}{2^n},\frac{k+1}{2^n}\Big],\qquad k\in\mathbb{Z}_{2^n}.
\]

\subsection{Quantum registers and one-qubit embeddings}\label{subsec:registers}

We consider the $n$-qubit Hilbert space
\[
\mathcal{H}_n:=(\mathbb{C}^2)^{\otimes n}\cong \mathbb{C}^{2^n},
\]
endowed with the computational basis $\{\,|b_n(k)\rangle: k\in\mathbb{Z}_{2^n}\,\}$.
Given a one-qubit unitary $U\in \mathrm{U}(2)$ and an index $1\le j\le n$, we define
the standard embedding (acting on the $j$-th qubit)
\[
w_j^{(n)}(U):=\rI_2^{\otimes (j-1)}\otimes U\otimes \rI_2^{\otimes (n-j)}\in \mathrm{U}(2^n).
\]

\subsection{Selective and controlled one-qubit operations}\label{subsec:controlled}

Fix $1\le \ell\le n$ and $U\in\mathrm{U}(2)$. For a pattern
$\mathbf{z}=z_1\cdots z_{\ell-1}z_{\ell+1}\cdots z_n\in\mathbb{Z}_2^{n-1}$,
we define the selective (rank-one supported) operator
\[
\rC^{(\ell)}_{\mathbf{z}}(U)
:=
|z_1\cdots z_{\ell-1}\rangle\langle z_1\cdots z_{\ell-1}|
\otimes U\otimes
|z_{\ell+1}\cdots z_n\rangle\langle z_{\ell+1}\cdots z_n|.
\]
This operator acts as $w_\ell^{(n)}(U)$ on basis states whose $(n-1)$ remaining bits
match $\mathbf{z}$, and annihilates all other basis states. From $\rC^{(\ell)}_{\mathbf{z}}(U)$
we build the corresponding controlled unitary (identity outside the selected branch)
\[
\CC^{(\ell)}_{\mathbf{z}}(U)
:=
\rC^{(\ell)}_{\mathbf{z}}(U)+
\sum_{\mathbf{z}'\in\mathbb{Z}_2^{n-1}\setminus\{\mathbf{z}\}}
\rC^{(\ell)}_{\mathbf{z}'}(\rI_2)
\;\in\;\mathrm{U}(2^n).
\]
In particular, $\CC^{(\ell)}_{\mathbf{z}}(U)$ applies $U$ to qubit $\ell$ if and only if the other
qubits match the control pattern $\mathbf{z}$, and otherwise acts as the identity.

\subsection{Rotation primitives and amplitude encoding objective}\label{subsec:rotations}

We will use the real rotation gate $\rR(\alpha)\in\mathrm{U}(2)$ defined by
\[
\rR(\alpha):=
\begin{pmatrix}
\cos\alpha & -\sin\alpha\\
\sin\alpha & \cos\alpha
\end{pmatrix},
\qquad \alpha\in\mathbb{R}.
\]
Given a nonnegative density $\varrho:[0,1]\to[0,\infty)$ with $\int_0^1\varrho(x)\,dx=1$,
we define the target probabilities
\[
p_k:=\int_{I_k}\varrho(x)\,dx,\qquad k\in\mathbb{Z}_{2^n}.
\]
The state-preparation (amplitude-encoding) objective is to construct a circuit $U\in\mathrm{U}(2^n)$
such that
\[
U|0\rangle^{\otimes n}=\sum_{k\in\mathbb{Z}_{2^n}}\sqrt{p_k}\,|b_n(k)\rangle,
\]
up to a global phase.

\subsection{Elementary identities for dyadic refinement}\label{subsec:binary-dyadic}

The binary representation $b_m$ and the dyadic grid points and intervals
$z^{(\ell)}_{b_\ell(k)}$, $I^{(\ell)}_{b_\ell(k)}$ are as defined in
Section~\ref{subsec:dyadic}; here we extend the index $m$ to an arbitrary positive integer.
The following two lemmas collect the shift and refinement identities used
throughout the proof of the main theorem.

\begin{lemma}[Binary shift identities]\label{lem:binary-shift}
Let $1\le \ell\le n-1$.
\begin{enumerate}
\item[\rm(i)] For $k\in \mathbb{Z}_{2^{\ell}}$:
\[
b_{\ell+1}(2k)=0\,b_\ell(k),\qquad
b_{\ell+1}(2k+1)=1\,b_\ell(k).
\]
\item[\rm(ii)] For $k\in \mathbb{Z}_{2^{\ell}}$ where $k \le 2^{\ell}-2:$
\[
b_{\ell+1}(2k+2)=0\,b_\ell(k+1).
\]
\end{enumerate}
\end{lemma}

\begin{proof}
Write $k=\sum_{j=1}^\ell z_j2^{j-1}$, i.e.\ $b_\ell(k)=z_1\cdots z_\ell$.
Then $2k=\sum_{j=1}^\ell z_j2^{j}$, hence the length-$(\ell+1)$ binary expansion of
$2k$ is $0z_1\cdots z_\ell=0\,b_\ell(k)$, proving the first identity in~(i).
Since $2k+1=1+\sum_{j=1}^\ell z_j2^{j}$, the leading bit is $1$ and the remaining
bits reproduce $b_\ell(k)$, giving $b_{\ell+1}(2k+1)=1\,b_\ell(k)$.
For~(ii), the hypothesis $k\le 2^\ell-2$ ensures $k+1\le 2^\ell-1$, so $b_\ell(k+1)$
is well-defined; the identity $b_{\ell+1}(2k+2)=0\,b_\ell(k+1)$ then follows by
applying the first identity in~(i) with $k$ replaced by $k+1$.
\end{proof}

\begin{corollary}[Dyadic refinement]\label{cor:dyadic-refinement}
Let $1\le \ell\le n-1$ and $k\in \mathbb{Z}_{2^{\ell}}$. Then
\[
I^{(\ell)}_{b_\ell(k)}
=
I^{(\ell+1)}_{\,b_{\ell+1}(2k)}\;\cup\; I^{(\ell+1)}_{\,b_{\ell+1}(2k+1)}
=
I^{(\ell+1)}_{\,0b_\ell(k)}\;\cup\; I^{(\ell+1)}_{\,1b_\ell(k)},
\]
and the union is disjoint up to the common endpoint.
Consequently, for any integrable $\varrho$,
\[
\int_{I^{(\ell)}_{b_\ell(k)}}\varrho(x)\,dx
=
\int_{I^{(\ell+1)}_{0b_\ell(k)}}\varrho(x)\,dx
+
\int_{I^{(\ell+1)}_{1b_\ell(k)}}\varrho(x)\,dx.
\]
\end{corollary}

\begin{lemma}[Suffix extraction via scaling]\label{lem:suffix-extraction}
Let $k\in \mathbb{Z}_{2^n}$ with $b_n(k)=z_1z_2\cdots z_n$.
For $1\le s\le n-1$ we have
\[
\Big\lfloor 2^{-s}k\Big\rfloor
=
\sum_{j=s+1}^n z_j\,2^{j-s-1},
\qquad\text{and hence}\qquad
b_{n-s}\!\Big(\lfloor 2^{-s}k\rfloor\Big)=z_{s+1}\cdots z_n.
\]
\end{lemma}

\begin{proof}
Using $k=\sum_{j=1}^n z_j2^{j-1}$,
\[
2^{-s}k=\sum_{j=1}^s z_j2^{j-s-1}+\sum_{j=s+1}^n z_j2^{j-s-1}.
\]
The first sum belongs to $[0,1)$, so taking the integer part yields the second sum,
which is exactly the claimed binary value.
\end{proof}

\subsection{Measurement probabilities and statement of the main theorem}\label{subsec:born-main}

We briefly recall how probabilities are extracted from a quantum state, since this is the
criterion used to certify the correctness of the Grover--Rudolph construction.

A (mixed) quantum state is represented by a density matrix
\[
\rho \in \mathbb{M}_N(\mathbb{C}),\qquad \rho\succeq 0,\qquad \tr(\rho)=1,
\]
and pure states correspond to rank-one projectors $\rho=\ket{\psi}\bra{\psi}$ with $\|\psi\|=1$.
Given a unitary $U\in\mathrm{U}(N)$ and an input state $\rho_0$, the output state is
\[
\rho = U\rho_0 U^\star.
\]
For the computational-basis measurement, define the orthogonal projectors
\[
\Pi_k := \ket{b_n(k)}\bra{b_n(k)},\qquad k\in\mathbb{Z}_{2^n},
\qquad
\sum_{k\in\mathbb{Z}_{2^n}}\Pi_k = I_N.
\]
The associated outcome is a discrete random variable $\rA$ taking values in $\mathbb{Z}_{2^n}$, and
the Born rule yields the probability law
\begin{equation}\label{eq:born}
\mathbb{P}_{\rho}(\rA=k) = \tr(\rho\,\Pi_k),\qquad k\in\mathbb{Z}_{2^n}.
\end{equation}
In particular, if $\rho=\ket{\psi}\bra{\psi}$ is pure then
\[
\mathbb{P}_{\rho}(\rA=k) = |\braket{b_n(k)}{\psi}|^2.
\]

\begin{theorem}[Grover--Rudolph state preparation]\label{theorem:Grover-Rudolph}
Let $n\ge 1$ and set $N:=2^n$. Let $\varrho:[0,1]\to[0,\infty)$ satisfy $\int_0^1\varrho(x)\,dx=1$
and define $\{p_k\}_{k\in\mathbb{Z}_{2^n}}$ by $p_k=\int_{I_k}\varrho(x)\,dx$.
Then there exists a quantum circuit $U\in\mathrm{U}(N)$, constructed as a product of
multi-controlled single-qubit rotations, such that for $\rho_0=\ket{0}^{\otimes n}\bra{0}^{\otimes n}$
and $\rho:=U\rho_0U^\star$ one has
\[
\mathbb{P}_{\rho}(\rA=k)=\tr(\rho\,\Pi_k)=p_k,\qquad \forall\,k\in\mathbb{Z}_{2^n}.
\]
Moreover, the circuit can be chosen with $N-1$ elementary controlled-rotation gates.
\end{theorem}

\section{Proof of the Grover--Rudolph Theorem}\label{sec:proof_GR}

Throughout this section we set $N=2^n$ and we work in $\mathbb{M}_N(\mathbb{C})$ with the
computational basis $\{\ket{b_n(k)}:k\in\mathbb{Z}_{2^n}\}$, initial state
$\rho_0=\ket{b_n(0)}\bra{b_n(0)}=\ket{0}^{\otimes n}\bra{0}^{\otimes n}$, and the measurement
random variable $\rA$ associated with the projectors $\ket{b_n(k)}\bra{b_n(k)}$.
For a bit $z\in\mathbb{Z}_2$ and an angle $x\in\mathbb{R}$ we use
\[
\mathsf{T}_z(x) := (\cos x)^{1-z}(\sin x)^z .
\]

\subsection{Dyadic integrals and trigonometric factorization}\label{subsec:factorization}

Let $\varrho:[0,1]\to[0,\infty)$ be a density, i.e.\ $\int_0^1 \varrho(x)\,dx=1$.
For $m\ge 1$ and a binary word $\mathbf{w}\in\mathbb{Z}_2^{m}$, let $k(\mathbf{w})\in\mathbb{Z}_{2^m}$ be the unique
integer such that $b_m(k(\mathbf{w}))=\mathbf{w}$ (equivalently, if $\mathbf{w}=z_1\cdots z_m$ then $k(\mathbf{w})=\sum_{r=1}^{m} z_r2^{r-1}$).
We associate to $\mathbf{w}$ the dyadic interval of length $2^{-m}$
\[
I_{\mathbf{w}}:=\Big[\frac{k(\mathbf{w})}{2^{m}},\,\frac{k(\mathbf{w})+1}{2^{m}}\Big]\subset[0,1],
\]
and define its probability mass by
\[
p_{\mathbf{w}}:=\int_{I_w}\varrho(x)\,dx.
\]
In particular, if $k\in\mathbb{Z}_{2^n}$ and $b_n(k)=\mathbf{w}=z_1\cdots z_n$, then $k(\mathbf{w})=k$ and thus
\[
p_{z_1\cdots z_n}=p_{\mathbf{w}}=\int_{I_w}\varrho(x)\,dx=\int_{k/2^n}^{(k+1)/2^n}\varrho(x)\,dx.
\]
By Corollary~\ref{cor:dyadic-refinement}, every dyadic interval indexed by a word
$\mathbf{w}\in\mathbb{Z}_2^{m}$ with $m<n$ splits into its two children $0\mathbf{w}$ and $1\mathbf{w}$. Therefore, for the associated
probability masses we have the refinement identity
\[
p_{\mathbf{w}}=\int_{I^{(m)}_{w}}\varrho(x)\,dx
=\int_{I^{(m+1)}_{0w}}\varrho(x)\,dx+\int_{I^{(m+1)}_{1w}}\varrho(x)\,dx
= p_{0\mathbf{w}}+p_{1\mathbf{w}}.
\]

\medskip
\noindent\textbf{Angle assignment.}
For each word \(\mathbf w\in\mathbb Z_2^m\), \(0\le m\le n-1\), with
\(p_{\mathbf w}>0\), we define \(\theta_{\mathbf w}\in[0,\pi/2]\) by
\[
\cos^2\theta_{\mathbf w}
=
\frac{p_{0\mathbf w}}{p_{\mathbf w}},
\qquad
\sin^2\theta_{\mathbf w}
=
\frac{p_{1\mathbf w}}{p_{\mathbf w}}.
\]
This is well defined because
\(p_{\mathbf w}=p_{0\mathbf w}+p_{1\mathbf w}\). If
\(p_{\mathbf w}=0\), then \(p_{0\mathbf w}=p_{1\mathbf w}=0\), and the
choice of \(\theta_{\mathbf w}\) is immaterial for the prepared state;
we set \(\theta_{\mathbf w}:=0\) for definiteness.

For the empty word \(\emptyset\), we write
\(\theta:=\theta_{\emptyset}\). This is the stage-\(1\) angle
corresponding to the top-level dyadic split, namely
\[
\theta
=
\theta_{\emptyset}
:=
\arccos
\sqrt{\frac{p_0}{p_{\emptyset}}}
=
\arccos
\sqrt{\int_0^{1/2}\varrho(x)\,dx},
\]
where \(p_{\emptyset}=\int_0^1\varrho(x)\,dx=1\).

\begin{proposition}[Trigonometric factorization]\label{prop:trig_factorization}
For each $k\in\mathbb{Z}_{2^n}$ with $b_n(k)=z_1\cdots z_n$ one has
\begin{equation}\label{eq:factorization}
p_{z_1\cdots z_n}
=
\mathsf{T}_{z_1}^2(\theta_{z_2\cdots z_n})\,
\mathsf{T}_{z_2}^2(\theta_{z_3\cdots z_n})\cdots
\mathsf{T}_{z_{n-1}}^2(\theta_{z_n})\,
\mathsf{T}_{z_n}^2(\theta).
\end{equation}
\end{proposition}

\begin{proof}
Fix $k\in\mathbb{Z}_{2^n}$ and write $b_n(k)=z_1\cdots z_n$.
For $r\in\{1,\dots,n\}$ let $\mathbf{w}_r:=z_r z_{r+1}\cdots z_n$ denote the length-$(n-r+1)$ suffix,
and set $\mathbf{w}_{n+1}:=\emptyset$. We prove by backward induction on $r$ that
\begin{equation}\label{eq:ind_claim_general}
p_{\mathbf{w}_r}
=
\mathsf{T}_{z_r}^2(\theta_{\mathbf{w}_{r+1}})\,
\mathsf{T}_{z_{r+1}}^2(\theta_{\mathbf{w}_{r+2}})\cdots
\mathsf{T}_{z_{n-1}}^2(\theta_{\mathbf{w}_n})\,
\mathsf{T}_{z_n}^2(\theta),
\end{equation}
where $\theta=\theta_{\emptyset}$ and, for any word $\mathbf{w}$, the angle $\theta_{\mathbf{w}}$ is defined by the convention
used in the manuscript (in particular, if $p_{\mathbf{w}}=0$ we set $\theta_{\mathbf{w}}:=0$).

\emph{Base case ($r=n$).}
We must show $p_{z_n}=\mathsf{T}_{z_n}^2(\theta)$. This holds by the definition of $\theta=\theta_{\emptyset}$:
\[
p_0=p_{\emptyset}\cos^2\theta=\cos^2\theta=\mathsf{T}_0^2(\theta),\qquad
p_1=p_{\emptyset}\sin^2\theta=\sin^2\theta=\mathsf{T}_1^2(\theta),
\]
since $p_{\emptyset}=\int_0^1\varrho(x)\,dx=1$.

\emph{Inductive step.}
Assume \eqref{eq:ind_claim_general} holds for $r+1$, and set $\mathbf{w}:=\mathbf{w}_{r+1}=z_{r+1}\cdots z_n$.
We prove \eqref{eq:ind_claim_general} for $r$. There are two cases.

\smallskip
\noindent\textbf{Case 1: $p_{\mathbf{w}}=0$.}
By Corollary~\ref{cor:dyadic-refinement} (Dyadic refinement), $p_{0\mathbf{w}}+p_{1\mathbf{w}}=p_{\mathbf{w}}=0$, and since
all probabilities are nonnegative it follows that $p_{0\mathbf{w}}=p_{1\mathbf{w}}=0$. In particular,
$p_{\mathbf{w}_r}=p_{z_r \mathbf{w}}=0$. On the other hand, our convention $\theta_{\mathbf{w}}=0$ implies that
$\mathsf{T}_{z_r}^2(\theta_{\mathbf{w}})\in\{0,1\}$, and the right-hand side of \eqref{eq:ind_claim_general} contains
the factor $\mathsf{T}_{z_r}^2(\theta_{\mathbf{w}})$ multiplying the remaining terms. But the remaining product
equals $p_{\mathbf{w}}$ by the induction hypothesis applied at level $r+1$ (since its left-hand side is $p_{\mathbf{w}}$),
and therefore the whole right-hand side equals $\mathsf{T}_{z_r}^2(\theta_{\mathbf{w}})\,p_{\mathbf{w}}=0$.
Hence \eqref{eq:ind_claim_general} holds.

\smallskip
\noindent\textbf{Case 2: $p_{\mathbf{w}}>0$.}
By definition of $\theta_{\mathbf{w}}$ we have
\[
p_{0\mathbf{w}}=p_{\mathbf{w}}\cos^2\theta_{\mathbf{w}}=p_{\mathbf{w}}\,\mathsf{T}_0^2(\theta_{\mathbf{w}}),
\qquad
p_{1\mathbf{w}}=p_{\mathbf{w}}\sin^2\theta_{\mathbf{w}}=p_{\mathbf{w}}\,\mathsf{T}_1^2(\theta_{\mathbf{w}}).
\]
Therefore, for $z_r\in\{0,1\}$,
\[
p_{\mathbf{w}_r}=p_{z_r \mathbf{w}}=p_{\mathbf{w}}\,\mathsf{T}_{z_r}^2(\theta_{\mathbf{w}}).
\]
Applying the induction hypothesis at level $r+1$ to $p_{\mathbf{w}}=p_{\mathbf{w}_{r+1}}$ yields
\[
p_{\mathbf{w}_r}
=\mathsf{T}_{z_r}^2(\theta_{\mathbf{w}_{r+1}})\,
\mathsf{T}_{z_{r+1}}^2(\theta_{\mathbf{w}_{r+2}})\cdots
\mathsf{T}_{z_{n-1}}^2(\theta_{\mathbf{w}_n})\,
\mathsf{T}_{z_n}^2(\theta),
\]
which is exactly \eqref{eq:ind_claim_general}.

\smallskip
This completes the induction. Taking $r=1$ proves \eqref{eq:factorization}.
\end{proof}

\subsection{Controlled rotations and state preparation}\label{subsec:controlled_rotations}

For $\alpha\in\mathbb{R}$ let
\[
\rR(\alpha):=\begin{bmatrix}\cos\alpha&-\sin\alpha\\ \sin\alpha&\cos\alpha\end{bmatrix}\in \mathrm{U}(2).
\]
For $j\in\{2,\dots,n\}$ and a binary word $\mathbf{w}\in\mathbb{Z}_2^{j-1}$, we denote by
\[
\CC_{\mathbf{w}}^{(1)}\bigl(\rR(\theta_{\mathbf{w}})\bigr)\in \mathrm{U}(2^j)
\]
the multi-controlled one-qubit rotation acting on the \emph{first} qubit, conditioned on the last
$(j-1)$ qubits being equal to $\mathbf{w}$, and acting as the identity on all other basis states. Equivalently,
$\CC_{\mathbf{w}}^{(1)}(\rR(\theta_{\mathbf{w}}))$ applies $\rR(\theta_{\mathbf{w}})$ to qubit $1$ if and only if the control register
(qubits $2,\dots,j$ in the $j$-qubit active subregister, or, in the full $n$-qubit register, the last
$(j-1)$ wires) matches the pattern $\mathbf{w}$.

Define the stage unitaries by
\[
\rU_1 := \rI_{2^{n-1}}\otimes \rR(\theta),
\qquad
\rU_j := \prod_{\mathbf{w}\in\mathbb{Z}_2^{j-1}}\Bigl(\rI_{2^{n-j}}\otimes \CC_{\mathbf{w}}^{(1)}\bigl(\rR(\theta_{\mathbf{w}})\bigr)\Bigr),
\quad 2\le j\le n.
\]
In particular, only the gates $\CC_{\mathbf{w}}^{(1)}(\rR(\theta_{\mathbf{w}}))$ (with a fixed target qubit) are required for
the inductive construction, while $\rU_1$ is an uncontrolled rotation on the last qubit.

Finally, we define the full circuit
\[
\rU := \rU_n\rU_{n-1}\cdots \rU_1 \in \mathrm{U}(2^n).
\]
By construction, $\rU_j$ is a product of $2^{j-1}$ elementary controlled-rotation gates, hence the
total number of elementary gates is
\[
\sum_{j=1}^n 2^{j-1} = 2^n-1 = N-1,
\]
i.e. $\rU$ has circuit length $N-1$.

\begin{lemma}[Inductive amplitude form]\label{lem:inductive_state}
Let $\ket{\psi_j}:=\rU_j\rU_{j-1}\cdots \rU_1\ket{0}^{\otimes n}$.
Then, for every $1\le j\le n$,
\[
\ket{\psi_j}
=
\ket{0}^{\otimes (n-j)}\otimes
\sum_{\mathbf{w}=z_{n-j+1}\cdots z_n\in\mathbb{Z}_2^{j}}
\Bigl(\mathsf{T}_{z_{n-j+1}}(\theta_{z_{n-j+2}\cdots z_n})\cdots
\mathsf{T}_{z_{n-1}}(\theta_{z_n})\mathsf{T}_{z_n}(\theta)\Bigr)\,
\ket{\mathbf{w}}.
\]
In particular,
\[
\ket{\psi_n}
=
\sum_{z_1\cdots z_n\in\mathbb{Z}_2^n}
\Bigl(\mathsf{T}_{z_1}(\theta_{z_2\cdots z_n})\cdots \mathsf{T}_{z_n}(\theta)\Bigr)\ket{z_1\cdots z_n}.
\]
\end{lemma}

\begin{proof}
We proceed by induction on $j$.

\smallskip
\noindent\textbf{Base case ($j=1$).}
By definition, $\rU_1=\rI_{2^{n-1}}\otimes \rR(\theta)$ acts only on the last qubit. Hence
\[
\ket{\psi_1}
=\rU_1\ket{0}^{\otimes n}
=\ket{0}^{\otimes(n-1)}\otimes \rR(\theta)\ket{0}.
\]
Since $\rR(\theta)\ket{0}=\cos\theta\,\ket{0}+\sin\theta\,\ket{1}$, we obtain
\[
\ket{\psi_1}
=\ket{0}^{\otimes(n-1)}\otimes\bigl(\cos\theta\,\ket{0}+\sin\theta\,\ket{1}\bigr)
=\ket{0}^{\otimes(n-1)}\otimes\sum_{z_n\in\mathbb{Z}_2}\mathsf{T}_{z_n}(\theta)\ket{z_n},
\]
which is the claimed formula for $j=1$.

\smallskip
\noindent\textbf{Inductive step.}
Assume the statement holds for some $j-1$ with $2\le j\le n$, i.e.
\begin{equation}\label{eq:psi_jminus1}
\ket{\psi_{j-1}}
=
\ket{0}^{\otimes (n-(j-1))}\otimes
\sum_{\mathbf{w}=z_{n-j+2}\cdots z_n\in\mathbb{Z}_2^{j-1}}
A_{\mathbf{w}}\,\ket{\mathbf{w}},
\end{equation}
where, for each $\mathbf{w}=z_{n-j+2}\cdots z_n$, the amplitude $A_{\mathbf{w}}$ is given by the product
\begin{equation}\label{eq:Aw_def}
A_{\mathbf{w}}
=
\mathsf{T}_{z_{n-j+2}}(\theta_{z_{n-j+3}\cdots z_n})\cdots
\mathsf{T}_{z_{n-1}}(\theta_{z_n})\,
\mathsf{T}_{z_n}(\theta).
\end{equation}

At this stage, the basis vectors $\ket{\mathbf{w}}$ label the computational states of the last $(j-1)$ qubits.
Their identification with suffixes of the global binary index follows from the fact that truncating
the first $(j-1)$ bits of an integer corresponds to extracting its suffix.
More precisely, if $k\in\mathbb{Z}_{2^n}$ has binary expansion $b_n(k)=z_1\cdots z_n$, then
\[
b_{n-(j-1)}\!\bigl(\lfloor 2^{-(j-1)}k\rfloor\bigr)=z_{n-j+2}\cdots z_n,
\]
as shown in Lemma~\ref{lem:suffix-extraction}.

It will be convenient to isolate the $(n-j+1)$-th qubit (the next qubit to be rotated at stage $j$).
Write the tensor product in \eqref{eq:psi_jminus1} as
\[
\ket{0}^{\otimes(n-(j-1))}
=\ket{0}^{\otimes(n-j)}\otimes \ket{0},
\]
so that
\begin{equation}\label{eq:psi_jminus1_rewrite}
\ket{\psi_{j-1}}
=
\ket{0}^{\otimes(n-j)}\otimes
\sum_{\mathbf{w}\in\mathbb{Z}_2^{j-1}} A_{\mathbf{w}}\,\bigl(\ket{0}\otimes\ket{\mathbf{w}}\bigr),
\end{equation}
where the factor $\ket{0}$ corresponds to qubit $(n-j+1)$ and $\ket{\mathbf{w}}$ to the last $(j-1)$ qubits.

\smallskip
\noindent\emph{Action of $\rU_j$ on basis branches.}
By definition,
\[
\rU_j=\prod_{\mathbf{w}\in\mathbb{Z}_2^{j-1}}\Bigl(\rI_{2^{n-j}}\otimes \CC_{\mathbf{w}}^{(1)}\bigl(\rR(\theta_{\mathbf{w}})\bigr)\Bigr).
\]
Inside the active $j$-qubit register (qubits $(n-j+1),\dots,n$), the gate
$\CC_{\mathbf{w}}^{(1)}(\rR(\theta_{\mathbf{w}}))\in\mathrm{U}(2^j)$ acts as follows:
it applies $\rR(\theta_{\mathbf{w}})$ to the \emph{first} qubit of that register (i.e.\ qubit $(n-j+1)$)
if and only if the remaining $(j-1)$ qubits are in the computational basis state $\ket{\mathbf{w}}$, and
acts as the identity on all other computational basis states.
Consequently, for each $\mathbf{w}\in\mathbb{Z}_2^{j-1}$,
\begin{equation}\label{eq:control_action}
\CC_{\mathbf{w}}^{(1)}\bigl(\rR(\theta_{\mathbf{w}})\bigr)\,\bigl(\ket{0}\otimes\ket{\mathbf{w}}\bigr)
=\bigl(\rR(\theta_{\mathbf{w}})\ket{0}\bigr)\otimes\ket{\mathbf{w}},
\end{equation}
and for $\mathbf{w}'\neq \mathbf{w}$,
\[
\CC_{\mathbf{w}}^{(1)}\bigl(\rR(\theta_{\mathbf{w}})\bigr)\,\bigl(\ket{0}\otimes\ket{\mathbf{w}'}\bigr)
=\ket{0}\otimes\ket{\mathbf{w}'}.
\]
Since the projectors onto distinct control branches are orthogonal, these controlled rotations act on
disjoint subspaces and therefore commute; hence the product over $\mathbf{w}$ applies the correct rotation in each
branch independently.

\smallskip
\noindent\emph{Splitting of amplitudes.}
Applying $\rU_j$ to \eqref{eq:psi_jminus1_rewrite} and using \eqref{eq:control_action} gives
\begin{align*}
\ket{\psi_j}
&=\rU_j\ket{\psi_{j-1}}\\
&=
\ket{0}^{\otimes(n-j)}\otimes
\sum_{\mathbf{w}\in\mathbb{Z}_2^{j-1}} A_{\mathbf{w}}\,
\Bigl(\rR(\theta_{\mathbf{w}})\ket{0}\Bigr)\otimes\ket{\mathbf{w}}.
\end{align*}
Finally, $\rR(\theta_{\mathbf{w}})\ket{0}=\cos\theta_{\mathbf{w}}\,\ket{0}+\sin\theta_{\mathbf{w}}\,\ket{1}
=\sum_{z\in\mathbb{Z}_2}\mathsf{T}_{z}(\theta_{\mathbf{w}})\ket{z}$, so
\begin{align*}
\ket{\psi_j}
&=
\ket{0}^{\otimes(n-j)}\otimes
\sum_{\mathbf{w}\in\mathbb{Z}_2^{j-1}} A_{\mathbf{w}}\,
\sum_{z_{n-j+1}\in\mathbb{Z}_2}\mathsf{T}_{z_{n-j+1}}(\theta_{\mathbf{w}})\,
\ket{z_{n-j+1}}\otimes\ket{\mathbf{w}}\\
&=
\ket{0}^{\otimes(n-j)}\otimes
\sum_{z_{n-j+1}\in\mathbb{Z}_2}\;
\sum_{\mathbf{w}\in\mathbb{Z}_2^{j-1}}
\Bigl(A_{\mathbf{w}}\,\mathsf{T}_{z_{n-j+1}}(\theta_{\mathbf{w}})\Bigr)\,
\ket{z_{n-j+1}\mathbf{w}}.
\end{align*}
Renaming the concatenated word $z_{n-j+1}\mathbf{w}$ as
\[
\mathbf{w}'=z_{n-j+1}z_{n-j+2}\cdots z_n\in\mathbb{Z}_2^j,
\]
and substituting $A_{\mathbf{w}}$ from \eqref{eq:Aw_def}, we obtain exactly the claimed product form
\[
\ket{\psi_j}
=
\ket{0}^{\otimes (n-j)}\otimes
\sum_{\mathbf{w}'=z_{n-j+1}\cdots z_n\in\mathbb{Z}_2^{j}}
\Bigl(\mathsf{T}_{z_{n-j+1}}(\theta_{z_{n-j+2}\cdots z_n})\cdots
\mathsf{T}_{z_{n-1}}(\theta_{z_n})\mathsf{T}_{z_n}(\theta)\Bigr)\,
\ket{\mathbf{w}'}.
\]
This completes the induction. The final statement for $\ket{\psi_n}$ is the special case $j=n$.
\end{proof}

\subsection{Completion of the proof of Theorem~\ref{theorem:Grover-Rudolph}}\label{subsec:completion}

\begin{proof}[Proof of Theorem~\ref{theorem:Grover-Rudolph}]
Let $\rho:=\rU\rho_0\rU^\star=\ket{\psi_n}\bra{\psi_n}$, with $\ket{\psi_n}$ as in
Lemma~\ref{lem:inductive_state}. For $k\in\mathbb{Z}_{2^n}$ with $b_n(k)=z_1\cdots z_n$ we compute
\[
\mathbb{P}_{\rho}(\rA=k)
=\tr\bigl(\rho\,\ket{b_n(k)}\bra{b_n(k)}\bigr)
=|\braket{b_n(k)}{\psi_n}|^2
=\mathsf{T}_{z_1}^2(\theta_{z_2\cdots z_n})\cdots \mathsf{T}_{z_n}^2(\theta).
\]
By Proposition~\ref{prop:trig_factorization}, the right-hand side equals
$p_{z_1\cdots z_n}=\int_{k/2^n}^{(k+1)/2^n}\varrho(x)\,dx$, hence
\[
\mathbb{P}_{\rho}(\rA=k)=\int_{k/2^n}^{(k+1)/2^n}\varrho(x)\,dx,
\qquad \forall\,k\in\mathbb{Z}_{2^n}.
\]
Finally, the circuit length of $\rU$ is $N-1$ by the gate count above. This proves the theorem.
\end{proof}

\section{Stability under imperfect rotation angles}
\label{sec:stability}
The results above establish exact correctness of the Grover--Rudolph
construction when all rotation angles are implemented exactly. In
practice, however, the angles may be affected by finite-precision
rounding, synthesis error, or calibration error. In this section we
quantify how such deterministic angular perturbations propagate to the
final output distribution.

Throughout this section we keep the rotation convention used in the
previous sections, namely
\[
\rR(\alpha)
=
\begin{pmatrix}
\cos\alpha&-\sin\alpha\\
\sin\alpha&\cos\alpha
\end{pmatrix},
\qquad
\rR(\alpha)\ket{0}
=
\cos\alpha\,\ket{0}
+
\sin\alpha\,\ket{1}.
\]
Thus the Grover--Rudolph angle \(\theta_{\mathbf w}\) directly determines
the binary split
\[
\rho_{\mathbf w}(0)=\cos^2\theta_{\mathbf w},
\qquad
\rho_{\mathbf w}(1)=\sin^2\theta_{\mathbf w}.
\]

Let
\[
\Theta
=
\{\theta_{\mathbf w}:
\mathbf w\in\mathbb Z_2^m,\ 0\le m\le n-1\}
\]
be the exact family of Grover--Rudolph angles. We consider a perturbed
family
\[
\widetilde\Theta
=
\{\widetilde\theta_{\mathbf w}:
\mathbf w\in\mathbb Z_2^m,\ 0\le m\le n-1\},
\qquad
\widetilde\theta_{\mathbf w}
=
\theta_{\mathbf w}+\varepsilon_{\mathbf w}.
\]
Throughout this section we adopt the convention
\(\mathbb Z_2^0:=\{\emptyset\}\), so that the index set
\(\bigcup_{m=0}^{n-1}\mathbb Z_2^m\) includes the empty word \(\emptyset\)
at level \(m=0\); the corresponding angle \(\theta_\emptyset\) (resp.\
\(\widetilde\theta_\emptyset\), \(\varepsilon_\emptyset\)) is the
top-level rotation angle (resp.\ its perturbation, its error).
With this convention, all sums and maxima over
\(\mathbf w\in\mathbb Z_2^m\), \(0\le m\le n-1\), are well defined.
The perturbed circuit is obtained from the exact Grover--Rudolph circuit 
by replacing each occurrence of
\(\rR(\theta_{\mathbf w})\) by
\(\rR(\widetilde\theta_{\mathbf w})\), while keeping the same dyadic
tree and the same control structure.

For a binary word \(\mathbf z=z_1\cdots z_n\), we use the suffix notation
\[
\mathbf z_{r+1:n}:=z_{r+1}\cdots z_n,
\qquad
1\le r\le n-1,
\]
and we extend the convention to \(r=n\) by setting
\(\mathbf z_{n+1:n}:=\emptyset\), so that
\(\theta_{\mathbf z_{n+1:n}}=\theta_\emptyset\) is the top-level angle.
With these conventions, the
factorization proved in Proposition~\ref{prop:trig_factorization} reads
\[
p_{\mathbf z}
=
\prod_{r=1}^{n}
\mathsf T_{z_r}^2
\bigl(\theta_{\mathbf z_{r+1:n}}\bigr),
\]
where the suffix $\mathbf z_{r+1:n}$ has length $n-r \in \{0,1,\ldots,n-1\}$,
covering all control words $\mathbf w\in\mathbb Z_2^m$, $0\le m\le n-1$.
The perturbed output
distribution is therefore
\[
\widetilde p_{\mathbf z}
=
\prod_{r=1}^{n}
\mathsf T_{z_r}^2
\bigl(\widetilde\theta_{\mathbf z_{r+1:n}}\bigr).
\]

We measure the distance between probability distributions on
\(\mathbb Z_2^n\) using total variation distance,
\[
\mathrm{TV}(p,q)
:=
\frac12
\sum_{\mathbf z\in\mathbb Z_2^n}
|p_{\mathbf z}-q_{\mathbf z}|.
\]

\begin{theorem}[Stability under angular perturbations]
\label{thm:gr_angle_stability}
Let \(p\) be the probability distribution prepared by the exact
Grover--Rudolph circuit, and let \(\widetilde p\) be the probability
distribution obtained by replacing each angle
\(\theta_{\mathbf w}\) by
\[
\widetilde\theta_{\mathbf w}
=
\theta_{\mathbf w}
+
\varepsilon_{\mathbf w}.
\]
Then
\[
\mathrm{TV}(p,\widetilde p)
\le
\min\!\left(1,\;
\sum_{m=0}^{n-1}
\max_{\mathbf w\in\mathbb Z_2^m}
|\varepsilon_{\mathbf w}|\right).
\]
In particular, if
\[
|\varepsilon_{\mathbf w}|\le \eta
\qquad
\text{for all }
\mathbf w\in\mathbb Z_2^m,\ 0\le m\le n-1,
\]
then
\[
\mathrm{TV}(p,\widetilde p)
\le
\min(1,\,n\eta).
\]
\end{theorem}

\begin{proof}
For each word \(\mathbf w\in\mathbb Z_2^m\), \(0\le m\le n-1\), define
the exact binary transition probability
\[
\rho_{\mathbf w}(0)
=
\cos^2\theta_{\mathbf w},
\qquad
\rho_{\mathbf w}(1)
=
\sin^2\theta_{\mathbf w},
\]
and the perturbed binary transition probability
\[
\widetilde\rho_{\mathbf w}(0)
=
\cos^2\widetilde\theta_{\mathbf w},
\qquad
\widetilde\rho_{\mathbf w}(1)
=
\sin^2\widetilde\theta_{\mathbf w}.
\]
Thus, for \(\mathbf z=z_1\cdots z_n\), using the convention
\(\mathbf z_{n+1:n}:=\emptyset\),
\[
p_{\mathbf z}
=
\prod_{r=1}^{n}
\rho_{\mathbf z_{r+1:n}}(z_r),
\qquad
\widetilde p_{\mathbf z}
=
\prod_{r=1}^{n}
\widetilde\rho_{\mathbf z_{r+1:n}}(z_r).
\]

For \(s=0,\ldots,n\), define a hybrid distribution \(p^{(s)}\) by using
the perturbed transitions for the last \(s\) steps of the suffix
recursion and the exact transitions for the remaining steps:
\[
p^{(s)}_{\mathbf z}
=
\prod_{r=1}^{n-s}
\rho_{\mathbf z_{r+1:n}}(z_r)
\prod_{r=n-s+1}^{n}
\widetilde\rho_{\mathbf z_{r+1:n}}(z_r).
\]
Then \(p^{(0)}=p\) and \(p^{(n)}=\widetilde p\).
Each \(p^{(s)}\) is a probability distribution on \(\mathbb{Z}_2^n\): non-negativity
is clear since every factor \(\rho_{\mathbf{z}_{r+1:n}}(z_r)\) and
\(\widetilde\rho_{\mathbf{z}_{r+1:n}}(z_r)\) is non-negative, and the normalisation
\(\sum_{\mathbf{z}\in\mathbb{Z}_2^n}p^{(s)}_{\mathbf{z}}=1\) follows by summing
the product bit by bit: for each \(r\), the marginalisation
\(\sum_{z_r\in\{0,1\}}\rho_{\mathbf{z}_{r+1:n}}(z_r)=1\) (respectively
\(\sum_{z_r\in\{0,1\}}\widetilde\rho_{\mathbf{z}_{r+1:n}}(z_r)=1\)) collapses
the corresponding factor to \(1\), so the full sum telescopes to \(1\).
Hence, since \(\mathrm{TV}\) is a metric on the simplex of probability distributions,
the triangle inequality gives
\[
\mathrm{TV}(p,\widetilde p)
\le
\sum_{s=0}^{n-1}
\mathrm{TV}\bigl(p^{(s)},p^{(s+1)}\bigr).
\]

The two hybrid distributions \(p^{(s)}\) and \(p^{(s+1)}\) differ only in
the transition associated with words of length \(s\). More precisely,
conditioning on the suffix
\(\mathbf w\in\mathbb Z_2^s\), the only change is the replacement of the
binary distribution \(\rho_{\mathbf w}\) by
\(\widetilde\rho_{\mathbf w}\). Therefore,
\[
\mathrm{TV}\bigl(p^{(s)},p^{(s+1)}\bigr)
\le
\max_{\mathbf w\in\mathbb Z_2^s}
\mathrm{TV}
\bigl(\rho_{\mathbf w},\widetilde\rho_{\mathbf w}\bigr).
\]
For a binary distribution of the form
\[
\rho_{\mathbf w}
=
(\cos^2\theta_{\mathbf w},\sin^2\theta_{\mathbf w}),
\]
we have
\[
\mathrm{TV}
\bigl(\rho_{\mathbf w},\widetilde\rho_{\mathbf w}\bigr)
=
\left|
\cos^2\theta_{\mathbf w}
-
\cos^2\widetilde\theta_{\mathbf w}
\right|.
\]
Using
\[
\cos^2 a-\cos^2 b
=
-\sin(a+b)\sin(a-b),
\]
we obtain
\[
\left|
\cos^2\theta_{\mathbf w}
-
\cos^2\widetilde\theta_{\mathbf w}
\right|
\le
|\theta_{\mathbf w}-\widetilde\theta_{\mathbf w}|
=
|\varepsilon_{\mathbf w}|.
\]
Consequently,
\[
\mathrm{TV}\bigl(p^{(s)},p^{(s+1)}\bigr)
\le
\max_{\mathbf w\in\mathbb Z_2^s}
|\varepsilon_{\mathbf w}|.
\]
Summing over \(s=0,\ldots,n-1\) gives
\[
\mathrm{TV}(p,\widetilde p)
\le
\min\!\left(1,\,
\sum_{s=0}^{n-1}
\max_{\mathbf w\in\mathbb Z_2^s}
|\varepsilon_{\mathbf w}|\right).
\]
The uniform estimate follows immediately.
\end{proof}

\begin{corollary}[Finite-precision angle synthesis]
\label{cor:finite_precision_angles}
Assume that each Grover--Rudolph angle is approximated by an angle
\(\widetilde\theta_{\mathbf w}\) satisfying
\[
|\widetilde\theta_{\mathbf w}-\theta_{\mathbf w}|
\le
\eta.
\]
Then the output distribution \(\widetilde p\) of the perturbed circuit
satisfies
\[
\mathrm{TV}(p,\widetilde p)
\le
\min(1,\,n\eta).
\]
In particular, if the angles are rounded to a uniform grid of mesh
\(\Delta\), so that
\[
|\widetilde\theta_{\mathbf w}-\theta_{\mathbf w}|
\le
\frac{\Delta}{2},
\]
then
\[
\mathrm{TV}(p,\widetilde p)
\le
\min\!\left(1,\,\frac{n\Delta}{2}\right).
\]
\end{corollary}

\begin{remark}[Physical \(\RY\)-angles]
\label{rem:physical_RY_angles}
The previous estimates are stated in terms of the Grover--Rudolph
rotation convention
\[
\rR(\theta)
=
\begin{pmatrix}
\cos\theta&-\sin\theta\\
\sin\theta&\cos\theta
\end{pmatrix}.
\]
Many quantum-computing libraries use instead the standard convention
\[
\RY(\varphi)
=
\begin{pmatrix}
\cos(\varphi/2)&-\sin(\varphi/2)\\
\sin(\varphi/2)&\cos(\varphi/2)
\end{pmatrix}.
\]
Thus
\[
\rR(\theta)=\RY(2\theta).
\]
If the implemented physical angle is
\[
\widetilde\varphi_{\mathbf w}
=
\varphi_{\mathbf w}
+
\delta_{\mathbf w},
\qquad
\varphi_{\mathbf w}=2\theta_{\mathbf w},
\]
then the corresponding Grover--Rudolph angular perturbation is
\[
\varepsilon_{\mathbf w}
=
\frac{\delta_{\mathbf w}}{2}.
\]
Therefore Theorem~\ref{thm:gr_angle_stability} gives
\[
\mathrm{TV}(p,\widetilde p)
\le
\min\!\left(1,\;
\frac12
\sum_{m=0}^{n-1}
\max_{\mathbf w\in\mathbb Z_2^m}
|\delta_{\mathbf w}|\right),
\]
and, under the uniform physical-angle bound
\(|\delta_{\mathbf w}|\le\delta\),
\[
\mathrm{TV}(p,\widetilde p)
\le
\min\!\left(1,\,\frac{n\delta}{2}\right).
\]
\end{remark}

\begin{remark}[Interpretation of the bound]
\label{rem:stability_interpretation}
The estimate in Theorem~\ref{thm:gr_angle_stability} exploits the
hierarchical structure of the Grover--Rudolph construction. A direct
operator-norm perturbation bound for a product of \(2^n-1\) rotations
would lead to a worst-case estimate proportional to the total number of
rotations. By contrast, the present argument works at the level of the
conditional probability tree and shows that the untruncated perturbation
sum grows at most linearly with the depth \(n\) of the tree; the
\(\min(1,\cdot)\) truncation then reflects the fact that
\(\mathrm{TV}\le 1\) always holds, so the bound is non-trivial only when
\(n\eta<1\).
\end{remark}

\section{Numerical validation and sensitivity analysis}\label{sec:numerics}
This section provides the numerical validation of the Grover--Rudolph construction.
Section~\ref{subsec:numerical_example} illustrates Theorem~\ref{theorem:Grover-Rudolph} on a
three-qubit example, comparing the target probabilities obtained from the density $\varrho$
with the empirical frequencies produced by a circuit simulation.
Section~\ref{subsec:sensitivity} then carries out a systematic sensitivity study that
quantifies the total-variation error as a joint function of angle-quantization precision and
measurement shot count, for $n\in\{2,3,4\}$ qubits, thereby validating the theoretical
bounds established in Section~\ref{sec:stability}.

\subsection{A numerical example}\label{subsec:numerical_example}

Consider the piecewise-linear probability density (Figure~\ref{fig:Density_Function})
\[
\varrho(x) =
\begin{cases}
4x, & 0 \leq x \leq \frac{1}{2}, \\[2mm]
4 - 4x, & \frac{1}{2} \leq x \leq 1,
\end{cases}
\qquad\text{so that}\qquad \int_0^1 \varrho(x)\,dx=1.
\]
We set $n=3$ (hence $N=2^3=8$) and define the target probabilities
\[
p_k:=\int_{k/2^3}^{(k+1)/2^3} \varrho(x)\,dx,\qquad k\in\mathbb{Z}_{2^3}.
\]
By Theorem~\ref{theorem:Grover-Rudolph}, there exists a circuit $U\in\mathrm{U}(8)$, built from
multi-controlled one-qubit rotations, such that for $\rho_0=\ket{0}^{\otimes 3}\bra{0}^{\otimes 3}$
and $\rho=U\rho_0U^\star$ one has
\[
\mathbb{P}_{\rho}(\rA=k)=p_k,\qquad \forall k\in\mathbb{Z}_{2^3},
\]
where $\rA$ denotes the computational-basis measurement random variable.
Moreover, the construction uses $2^3-1=7$ controlled rotations, hence it is fully determined by
$7$ angles.

\begin{figure}[h]
    \centering
    \includegraphics[scale=0.5]{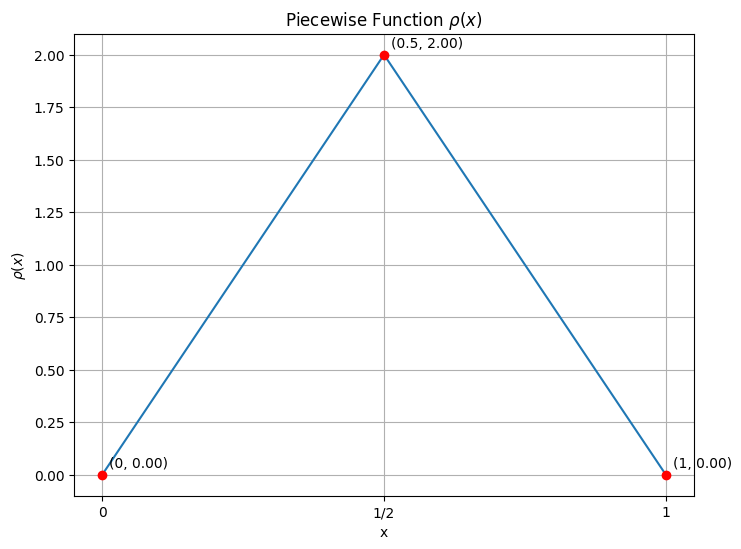}
    \caption{The probability density function $\varrho(x)$.}
    \label{fig:Density_Function}
\end{figure}

\paragraph{Step 1: computing $U_1$.}
The first stage acts on the last qubit and is given by
\[
\rU_1 := \rI_{2^{2}} \otimes \rR(\theta),
\qquad
\theta=\arccos\sqrt{\frac{\int_0^{1/2}\varrho(x)\,dx}{\int_0^1\varrho(x)\,dx}}
=\arccos\sqrt{\int_0^{1/2}\varrho(x)\,dx}=\frac{\pi}{4}.
\]
Figure~\ref{Figura_1} illustrates the integrals involved in the evaluation of $\theta$.

\begin{figure}[h]
    \centering
    \includegraphics[width=0.45\textwidth]{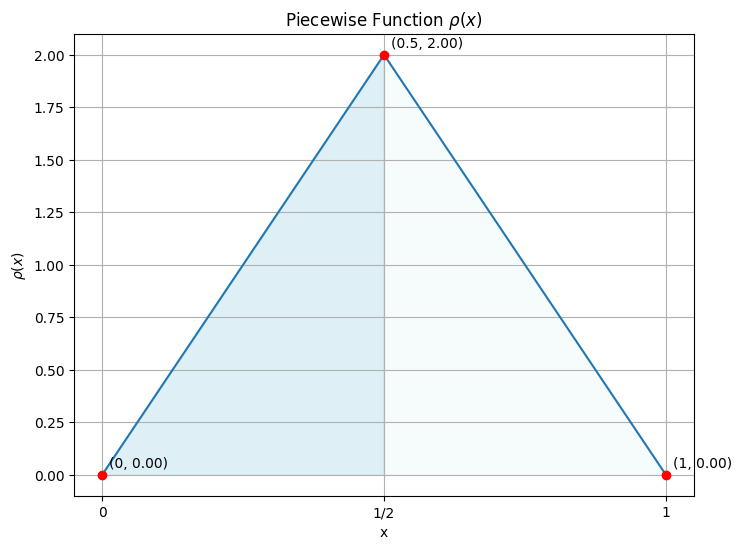}
    \caption{The integral $\int_{0}^{1} \varrho(x)\,dx$ (blue) and $\int_{0}^{1/2} \varrho(x)\,dx$ (dark blue).}
    \label{Figura_1}
\end{figure}

\paragraph{Step 2: computing $U_2$.}
At the second stage, we apply two controlled rotations (one per branch), namely
\[
\rU_2 := \prod_{z_2\in\mathbb{Z}_2}\Bigl(\rI_{2} \otimes \CC^{(1)}_{z_2}\,\rR(\theta_{z_2})\Bigr),
\]
with parameters
\[
\theta_0 = \arccos\sqrt{\frac{\int_{0}^{1/4} \varrho(x)\,dx}{\int_{0}^{1/2} \varrho(x)\,dx}} =  \frac{\pi}{3},
\qquad
\theta_1 = \arccos\sqrt{\frac{\int_{1/2}^{3/4} \varrho(x)\,dx}{\int_{1/2}^{1} \varrho(x)\,dx}} = \frac{\pi}{6}.
\]
Figure~\ref{Figura_2} shows the corresponding decomposition of the integrals.

\begin{figure}[h]
    \centering
    \includegraphics[width=0.5\textwidth]{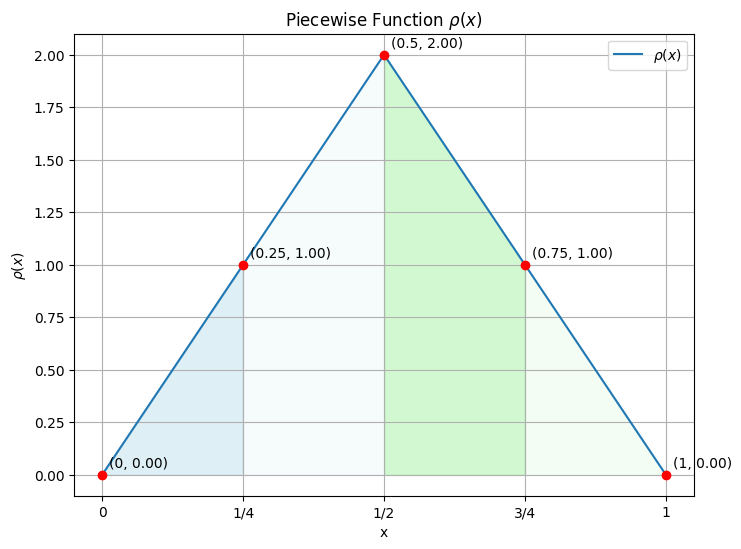}
    \caption{The integrals $\int_{0}^{1/2}\varrho(x)\,dx$ (blue) and $\int_{1/2}^{1}\varrho(x)\,dx$ (green), and their dyadic refinement.}
    \label{Figura_2}
\end{figure}

\paragraph{Step 3: computing $U_3$.}
Finally, we apply four controlled rotations corresponding to the four leaves of the depth-$2$
binary tree:
\[
\rU_3 := \prod_{z_1 z_2 \in \mathbb{Z}_{2}^2} \CC^{(1)}_{z_1 z_2}\,\rR(\theta_{z_1 z_2}),
\]
where the angles are
\[
\theta_{00} = \frac{\pi}{3},\qquad
\theta_{11} = \frac{\pi}{6},\qquad
\theta_{01} = \arccos \frac{\sqrt{21}}{6},\qquad
\theta_{10} = \arccos \frac{\sqrt{15}}{6}.
\]
Figure~\ref{Figura_3} illustrates the integral regions used to determine these parameters.
The complete circuit is therefore
\[
\rU := \rU_3\rU_2\rU_1
=
\left(\prod_{z_1 z_2 \in \mathbb{Z}_{2}^2} \CC^{(1)}_{z_1 z_2}\,\rR(\theta_{z_1 z_2})\right)
\left(\prod_{z_2\in\mathbb{Z}_2}\Bigl(\rI_{2} \otimes \CC^{(1)}_{z_2}\,\rR(\theta_{z_2})\Bigr)\right)
\left(\rI_{2^{2}} \otimes \rR(\theta)\right).
\]

\begin{figure}[h]
    \centering
    \includegraphics[width=0.55\textwidth]{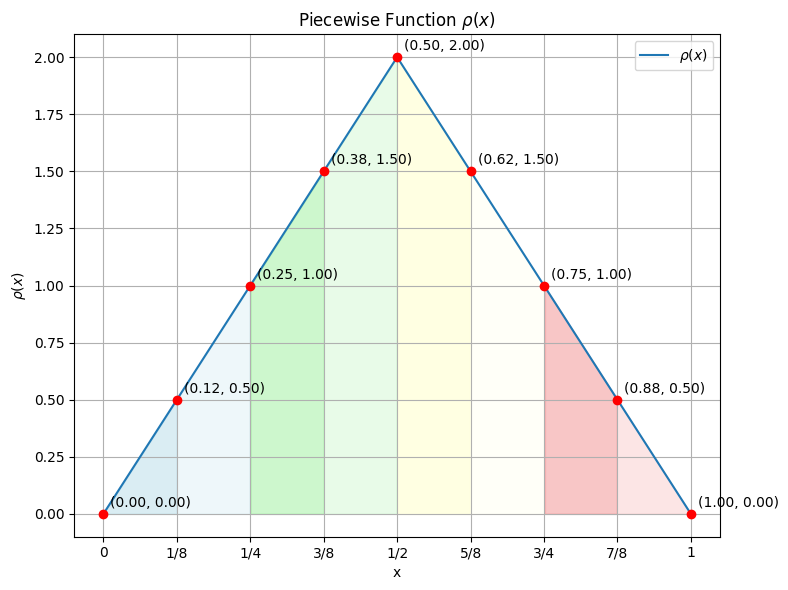}
    \caption{Integral regions used to compute $\theta_{00}$, $\theta_{01}$, $\theta_{10}$, and $\theta_{11}$.}
    \label{Figura_3}
\end{figure}

\paragraph{Quantum circuit implementation and simulation.}
Figure~\ref{fig:quantikz-circuit} displays the corresponding circuit.
We simulate this circuit using Qiskit's Aer module \texttt{qasm\_simulator}. Running $2048$ shots
on the simulated $3$-qubit device yields the empirical distribution shown in
Figure~\ref{Histogram}.

\begin{figure}[ht]
    \centering
    \scalebox{1.0}{
    \begin{quantikz}[row sep=0.5cm, font=\scriptsize]
  \lstick{$z_0$} &&& \gate[wires=3,style={fill=blue!10},label style={blue}]{U_3} & \\
  \lstick{$z_1$} && \gate[wires=2,style={fill=green!10},label style={customgreen}]{U_2} & & \\
  \lstick{$z_2$} & \gate[style={fill=red!10},label style={red}]{U_1} &&&
\end{quantikz}
    =
    \begin{quantikz}[row sep=0.5cm, font=\scriptsize]
  &&&&& \gate[style={fill=blue!10},label style={blue}]{\rR(\theta_{00})}
        & \gate[style={fill=blue!10},label style={blue}]{\rR(\theta_{10})}
        & \gate[style={fill=blue!10},label style={blue}]{\rR(\theta_{01})}
        & \gate[style={fill=blue!10},label style={blue}]{\rR(\theta_{11})}
        & \meter{} \\
  & \slice{} && \gate[style={fill=green!10},label style={customgreen}]{\rR(\theta_0)}
             & \gate[style={fill=green!10},label style={customgreen}]{\rR(\theta_1)}
             & \octrl{-1} & \ctrl{-1} & \octrl{-1} & \ctrl{-1} & \meter{} \\
  & \gate[style={fill=red!10},label style={red}]{\rR(\theta)}
    && \octrl{-1} & \ctrl{-1} \slice{}
    & \octrl{-2} & \octrl{-2} & \ctrl{-2} & \ctrl{-2} & \meter{}
\end{quantikz}
    }
    \caption{Quantum circuit $\rU$ implementing Theorem~\ref{theorem:Grover-Rudolph} for $n=3$.}
    \label{fig:quantikz-circuit}
\end{figure}
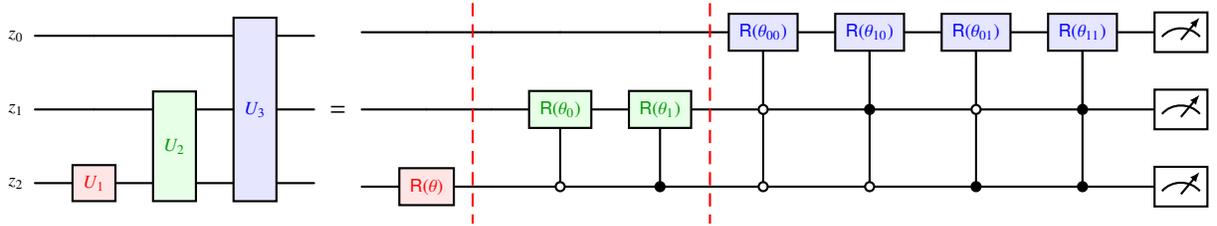

\begin{figure}[h]
    \centering
    \includegraphics[scale=0.5]{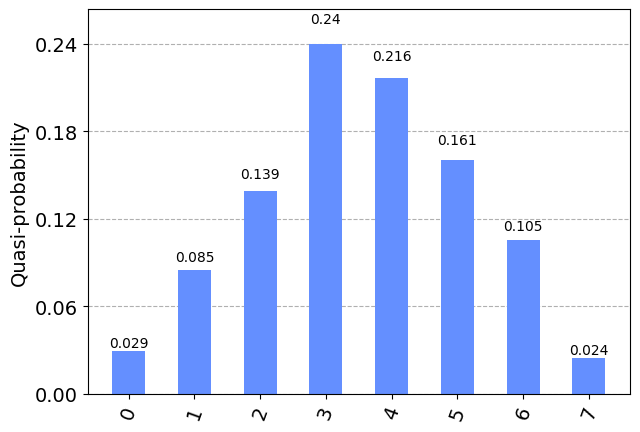}
    \caption{Empirical outcome distribution obtained from $2048$ shots of the simulated circuit.}
    \label{Histogram}
\end{figure}

\paragraph{Discussion.}
This example illustrates concretely how the Grover--Rudolph recursion translates dyadic integrals of
$\varrho$ into a hierarchy of controlled rotation angles.
Comparing the theoretical probabilities $\{p_k\}$ with the empirical frequencies produced by the simulator,
we observe an overall good agreement. In particular, the maximum deviation occurs on the interval
$[3/8,\,1/2]$, where the estimated probability is $0.240$ whereas the exact value is $0.21875$,
yielding an absolute error of approximately $0.02125$. The minimum deviation is observed on
$[1/2,\,5/8]$, with an absolute error of about $0.00275$.
These discrepancies are consistent with finite-shot sampling fluctuations and with the finite-precision
representation of rotation angles in the simulation.

A natural next step is to implement the same circuit on a real quantum device, where hardware noise,
gate calibration errors, and decoherence effects would provide a more realistic assessment of the method's
practical performance. Nevertheless, even in this idealized setting, the example provides a transparent
verification of the construction and supports the theoretical results established in the previous sections.

\subsection{Sensitivity to angle precision and shot noise}\label{subsec:sensitivity}
The Grover--Rudolph circuit depends on $N-1$ real-valued rotation angles.
In any practical implementation, two independent sources of error degrade the
output distribution: (i) \emph{angle quantization}, arising from finite-precision
representation of the $\theta_{\mathbf{w}}$ values, and (ii) \emph{shot noise}, arising
from finite sampling of the prepared quantum state.
We study both effects on the triangle density $\varrho$ for $n\in\{2,3,4\}$ qubits.

\paragraph{Experimental setup.}
Let $\{p_k\}_{k\in\mathbb{Z}_{2^n}}$ denote the exact target probabilities and let
$\{p_k^{(b)}\}$ denote the output probabilities obtained when each angle $\theta_{\mathbf{w}}$
is approximated by a finite-precision value.
As recalled in Remark~\ref{rem:physical_RY_angles}, the Grover--Rudolph angles
$\theta_{\mathbf{w}}\in[0,\pi/2]$ defined by $\cos^2\theta_{\mathbf{w}}=p_{0\mathbf{w}}/p_{\mathbf{w}}$
are passed to Qiskit as physical $\RY$ angles $\phi_{\mathbf{w}}:=2\theta_{\mathbf{w}}\in[0,\pi]$.
We quantise the physical angles $\phi_{\mathbf{w}}$ to the nearest multiple of
$\pi/2^{b-1}$ (a $b$-bit uniform grid on $[0,\pi]$), which is equivalent to quantising
$\theta_{\mathbf{w}}$ to the nearest multiple of $\pi/2^b$ on $[0,\pi/2]$.
We measure accuracy by the \emph{total variation distance} $\mathrm{TV}(p,q)$
defined in Section~\ref{sec:stability},
which equals half the $\ell^1$ error and is bounded in $[0,1]$.
All simulations use Qiskit's \texttt{AerSimulator}---the same backend employed in
Section~\ref{subsec:numerical_example}---with the \texttt{statevector} method for
angle-precision experiments (noise-free exact output) and the \texttt{automatic}
method (\texttt{qasm\_simulator}) for shot-noise experiments.
Shot-noise experiments average $10$ independent runs with fixed random seeds for reproducibility.

\paragraph{Effect of angle precision (Table~\ref{tab:tv_angles}).}
Table~\ref{tab:tv_angles} reports $\mathrm{TV}(p, p^{(b)})$ computed via statevector simulation
(no shot noise) for $b\in\{8,16,32\}$ bits and $n\in\{2,3,4\}$.
\begin{table}[h]
\centering
\begin{tabular}{c|ccc}
$n$ & $b=8$ & $b=16$ & $b=32$ \\
\hline
$2$ & $3.55\times10^{-3}$ & $1.38\times10^{-5}$ & $2.11\times10^{-10}$ \\
$3$ & $3.56\times10^{-3}$ & $1.66\times10^{-5}$ & $1.51\times10^{-10}$ \\
$4$ & $3.07\times10^{-3}$ & $8.80\times10^{-6}$ & $1.91\times10^{-10}$ \\
\end{tabular}
\caption{TV distance due to angle quantization alone
(\texttt{AerSimulator} statevector method, no shot noise).
Each entry is $\mathrm{TV}(p, p^{(b)})$ where $p^{(b)}$ uses a $b$-bit
uniform grid on the physical $\RY$ angle $\phi_{\mathbf{w}}=2\theta_{\mathbf{w}}\in[0,\pi]$,
equivalently a grid of step $\pi/2^b$ on $\theta_{\mathbf{w}}\in[0,\pi/2]$.}
\label{tab:tv_angles}
\end{table}
The TV error decreases rapidly with $b$: going from $8$ to $16$ bits reduces the error by
roughly three orders of magnitude, and $32$-bit precision yields errors below $10^{-9}$.
Notably, the error is nearly \emph{independent of $n$}: rounding each of the $2^n-1$
physical angles independently does not cause an accumulation of errors in TV distance
for fixed $b$, because each stage operates on an orthogonal subspace and errors do not
compound multiplicatively across stages in total variation.

\paragraph{Effect of shot noise (Table~\ref{tab:tv_shots}).}
Table~\ref{tab:tv_shots} reports the mean TV distance between $p$ and the empirical
frequency distribution obtained from $S\in\{256,1024,4096\}$ shots with exact angles,
simulated via the \texttt{AerSimulator} \texttt{qasm\_simulator} method.
\begin{table}[h]
\centering
\begin{tabular}{c|ccc}
$n$ & $S=256$ & $S=1024$ & $S=4096$ \\
\hline
$2$ & $0.03477$ & $0.01855$ & $0.01311$ \\
$3$ & $0.06250$ & $0.02559$ & $0.01733$ \\
$4$ & $0.09609$ & $0.04482$ & $0.02412$ \\
\end{tabular}
\caption{TV distance due to shot noise alone
(\texttt{AerSimulator} \texttt{qasm\_simulator}, exact angles,
mean over $10$ independent runs).}
\label{tab:tv_shots}
\end{table}
The TV distance scales approximately as $O(\sqrt{2^n/S})$, consistent with the
standard Hoeffding bound stated in Remark~\ref{rem:hoeffding_tv} below.
For fixed $S$, the error grows with $n$ because there are more outcomes to estimate;
for fixed $n$, the error decreases at the classical $O(1/\sqrt{S})$ rate.

\begin{remark}[Hoeffding bound for the TV distance]
\label{rem:hoeffding_tv}
Let $\hat p_k = n_k/S$ be the empirical frequency of outcome $k$ obtained from $S$ independent
measurements of the prepared state, so that $(n_0,\dots,n_{N-1})$ follows a multinomial
distribution with parameters $S$ and $(p_0,\dots,p_{N-1})$.
By the Hoeffding inequality~\cite{Hoeffding1963} applied to each coordinate
$\hat p_k - p_k\in[-1,1]$, together with the union bound over the $N=2^n$ outcomes,
one obtains that with probability at least $1-\delta$,
\[
\mathrm{TV}(p,\hat p)
=
\frac{1}{2}\sum_{k=0}^{N-1}|\hat p_k - p_k|
\;\le\;
\sqrt{\frac{2^n\log(2/\delta)}{2S}}.
\]
In particular, achieving $\mathrm{TV}(p,\hat p)\le\varepsilon$ with confidence $1-\delta$ requires
\[
S\;\ge\;\frac{2^n\log(2/\delta)}{2\varepsilon^2},
\]
i.e.\ $S=\Omega(2^n/\varepsilon^2)$ shots.
This should be contrasted with the deterministic angle-perturbation estimate
of Corollary~\ref{cor:finite_precision_angles}: for $b$-bit angle precision on
$\theta_{\mathbf{w}}\in[0,\pi/2]$ (equivalently, $b$-bit precision on the physical
$\RY$ angle $\phi_{\mathbf{w}}=2\theta_{\mathbf{w}}\in[0,\pi]$), the mesh is
$\Delta=\pi/2^b$ and the angle-induced TV error satisfies
$\mathrm{TV}(p,\tilde p)\le n\pi/2^{b+1}$, which is independent of $S$ and decreases
\emph{exponentially} in $b$.
The numerical results in Tables~\ref{tab:tv_angles}--\ref{tab:tv_shots} confirm
that the shot-noise term $O(\sqrt{2^n/S})$ dominates over the angle-precision term
$O(n/2^b)$ for all practically relevant values of $b\ge 8$ and $S\le 4096$.
\end{remark}
Table~\ref{tab:tv_combined} shows the combined TV distance for $b\in\{8,16,32\}$ bits
with varying shots; the results for $b=16$ and $b=32$ are numerically indistinguishable
at all shot counts, while $b=8$ shows a small but visible excess at low shot counts.
\begin{table}[h]
\centering
\begin{tabular}{cc|ccc}
& & $S=256$ & $S=1024$ & $S=4096$ \\
\hline
\multirow{3}{*}{$b=8$}
& $n=2$ & $0.03594$ & $0.01768$ & $0.01343$ \\
& $n=3$ & $0.06211$ & $0.02627$ & $0.01755$ \\
& $n=4$ & $0.09648$ & $0.04521$ & $0.02456$ \\
\hline
\multirow{3}{*}{$b=16$}
& $n=2$ & $0.03477$ & $0.01855$ & $0.01311$ \\
& $n=3$ & $0.06250$ & $0.02559$ & $0.01733$ \\
& $n=4$ & $0.09609$ & $0.04482$ & $0.02412$ \\
\hline
\multirow{3}{*}{$b=32$}
& $n=2$ & $0.03477$ & $0.01855$ & $0.01313$ \\
& $n=3$ & $0.06250$ & $0.02559$ & $0.01733$ \\
& $n=4$ & $0.09609$ & $0.04482$ & $0.02412$ \\
\end{tabular}
\caption{TV distance for combined angle quantization ($b$ bits) and shot noise
(\texttt{AerSimulator} \texttt{qasm\_simulator}, mean over $10$ runs).
Results for $b=16$ are numerically indistinguishable from $b=32$; the small
excess visible for $b=8$ at low shot counts reflects the angle-quantization
contribution $\mathrm{TV}(p,\tilde{p})\lesssim 4\times 10^{-3}$.}
\label{tab:tv_combined}
\end{table}

The combined results confirm that \emph{shot noise dominates} at practically relevant shot
counts: for $S\ge 256$, the TV distance with $b=8$ bits differs from the
corresponding value with $b=32$ bits by at most $0.00117$ in absolute terms
(e.g.\ $0.03594$ vs $0.03477$ for $n=2$, $S=256$), while the
shot-noise contribution alone ranges from $0.01311$ to $0.09648$ across the $n\in\{2,3,4\}$,
$S\in\{256,4096\}$ combinations.
In other words, for the triangle density with $n\le 4$ qubits, even $8$-bit angle
quantization introduces an angle-precision error (${\lesssim}\,4\times10^{-3}$ in TV) that is at least
one order of magnitude smaller than the shot-noise contribution at $S=256$.
This suggests that in practice, moderate angle precision ($b=8$ to $16$ bits) is entirely
sufficient, and the dominant cost of increasing $n$ is the number of shots $S=\Omega(2^n/\varepsilon^2)$
needed to achieve a given TV accuracy $\varepsilon$.

The following corollary combines Theorem~\ref{thm:gr_angle_stability} with the
Hoeffding bound of Remark~\ref{rem:hoeffding_tv} to give a single closed-form
guarantee on the total TV error from both sources simultaneously.

\begin{corollary}[Combined angle-quantization and shot-noise bound]
\label{cor:combined_bound}
Let $n\ge 1$, $\delta\in(0,1)$, and $\varepsilon\in(0,1]$.
Suppose each Grover--Rudolph angle $\theta_{\mathbf{w}}\in[0,\pi/2]$ is implemented
as a physical $\RY(2\widetilde\theta_{\mathbf{w}})$ gate, with
$|\widetilde\theta_{\mathbf{w}}-\theta_{\mathbf{w}}|\le\eta$ for some $\eta\ge 0$.
Let $\widetilde p$ denote the output distribution of the perturbed circuit, and let
$\hat p$ be the empirical frequency distribution obtained from $S$ independent
measurements of the perturbed circuit.
Then with probability at least $1-\delta$,
\begin{equation}\label{eq:combined_bound}
\mathrm{TV}(p,\hat p)
\;\le\;
\min\!\left(1,\;
n\eta
+
\sqrt{\frac{2^n\log(2/\delta)}{2S}}\right).
\end{equation}
In particular, if the angles are quantised to a uniform $b$-bit grid on $[0,\pi/2]$
(grid step $\pi/2^b$, maximum error $\eta=\pi/2^{b+1}$), then with probability at least $1-\delta$,
\begin{equation}\label{eq:combined_bound_bits}
\mathrm{TV}(p,\hat p)
\;\le\;
\min\!\left(1,\;
\frac{n\pi}{2^{b+1}}
+
\sqrt{\frac{2^n\log(2/\delta)}{2S}}\right).
\end{equation}
Consequently, for $\varepsilon\in(0,1]$ and $\delta\in(0,1)$, to achieve $\mathrm{TV}(p,\hat p)\le\varepsilon$ with confidence $1-\delta$,
it suffices to choose $b$ and $S$ such that
\begin{equation}\label{eq:design_rule}
\frac{n\pi}{2^{b+1}}\le\frac{\varepsilon}{2}
\qquad\text{and}\qquad
S\;\ge\;\frac{2^n\log(2/\delta)}{2(\varepsilon/2)^2}
=
\frac{2^{n+1}\log(2/\delta)}{\varepsilon^2}.
\end{equation}
The first condition requires $b\ge\log_2(2n\pi/\varepsilon)$, i.e.\ the precision
$b$ need only grow \emph{logarithmically} in $n$ and $1/\varepsilon$.
The second condition shows that the required shot count grows \emph{linearly}
in $N=2^n$ (the number of outcomes) for fixed $\varepsilon$ and $\delta$.
\end{corollary}

\begin{proof}
By the triangle inequality for total variation,
\[
\mathrm{TV}(p,\hat p)
\le
\mathrm{TV}(p,\widetilde p)
+
\mathrm{TV}(\widetilde p,\hat p).
\]
Theorem~\ref{thm:gr_angle_stability} gives
$\mathrm{TV}(p,\widetilde p)\le n\eta$.
The empirical distribution $\hat p$ is drawn from $S$ independent samples of
the perturbed circuit whose output distribution is $\widetilde p$; applying
Remark~\ref{rem:hoeffding_tv} with $\widetilde p$ in place of $p$ gives
$\mathrm{TV}(\widetilde p,\hat p)\le\sqrt{2^n\log(2/\delta)/(2S)}$ with
probability at least $1-\delta$.
Combining the two bounds yields~\eqref{eq:combined_bound}.
Inequality~\eqref{eq:combined_bound_bits} follows by substituting $\eta=\pi/2^{b+1}$.
The design rule~\eqref{eq:design_rule} is obtained by splitting the error budget
equally between the two terms and inverting each inequality.
\end{proof}

\begin{remark}[Separation of deterministic and stochastic errors]
\label{rem:separation}
Corollary~\ref{cor:combined_bound} makes explicit the qualitatively different
scaling of the two error sources.
The angle-precision term $n\eta$ is \emph{deterministic}, independent of $S$,
and decreases \emph{exponentially} in the bit depth $b$.
The shot-noise term $\sqrt{2^n\log(2/\delta)/(2S)}$ is \emph{stochastic},
independent of $b$, and decreases only as $1/\sqrt{S}$.
As a consequence, once $b$ is large enough that $n\pi/2^{b+1}$ is negligible
relative to $\varepsilon$, increasing $b$ further provides no benefit: the bottleneck
is entirely the shot budget.
The numerical experiments of Tables~\ref{tab:tv_angles}--\ref{tab:tv_combined}
confirm this picture for $n\in\{2,3,4\}$: already at $b=8$ bits the
angle-precision error (${\lesssim}\,4\times10^{-3}$) is dominated by the shot-noise
contribution at any practical shot count $S\le 4096$.
\end{remark}

\section{Ancilla-free transpilation of the Grover--Rudolph circuit in a $\{\RY(\cdot),X,\CNOT(\cdot\to\cdot)\}$ gate dictionary}
\label{sec:transpilation_triangulum}

In practical implementations, the circuit $\rU$ in
Theorem~\ref{theorem:Grover-Rudolph} must be expressed over a fixed elementary gate dictionary.
The choice of dictionary depends on the target hardware; here we focus on a gate set that is
widely supported by current quantum processors and quantum-computing simulators.
Throughout this section we work with
\[
\mathcal{G}:=\{\RY(\cdot),\,X,\,\CNOT(\cdot\to\cdot)\},
\]
where $\CNOT(c\to t)$ denotes the standard two-qubit controlled-NOT (i.e.\ a controlled-$X$ on the target $t$
conditioned on the control $c$ being in state $\ket{1}$).

\subsection{Active-register viewpoint and target gates}\label{subsec:transpilation_active}

At Grover--Rudolph stage $j\in\{2,\dots,n\}$ the construction acts nontrivially only on the last $j$ qubits of the
$n$-qubit register. We therefore introduce the \emph{active $j$-qubit register}
\[
(q_1,\dots,q_j):=(n-j+1,\;n-j+2,\;\dots,\;n),
\qquad\text{i.e.}\qquad q_r:=n-j+r,\ \ r=1,\dots,j,
\]
so that any $G\in\mathrm U(2^j)$ acting on the active register is embedded into $\mathcal H_n=(\mathbb C^2)^{\otimes n}$
as $\rI_{2^{n-j}}\otimes G\in\mathrm U(2^n)$.

For $\mathbf{w}\in\mathbb Z_2^{j-1}$ we write
\[
\CC_{\mathbf{w}}^{(1)}\bigl(\rR(\theta_{\mathbf{w}})\bigr)\in\mathrm U(2^j)
\]
for the pattern-controlled one-qubit rotation acting on the \emph{target} $q_1$ conditioned on the control register
$(q_2,\dots,q_j)$ being equal to $\mathbf{w}$, and acting as the identity on all other computational basis states.

\paragraph{Dictionary gates.}
For a qubit $q$ we denote by $X_q$ the Pauli-$X$ gate on $q$, and by $\RY(\alpha)_q$ the $y$-rotation
\[
\RY(\alpha)=\exp\!\Bigl(-\frac{i\alpha}{2}Y\Bigr)
=
\begin{pmatrix}
\cos(\alpha/2) & -\sin(\alpha/2)\\
\sin(\alpha/2) & \cos(\alpha/2)
\end{pmatrix}.
\]
For qubits $c$ (control) and $t$ (target), the standard $\CNOT(c\to t)$ is defined by
\[
\CNOT(c\to t)\ket{x}_c\ket{y}_t=\ket{x}_c\ket{y\oplus x}_t,
\qquad x,y\in\mathbb Z_2,
\]
and acts as the identity on all other qubits.

\paragraph{CNOT in terms of fully pattern-controlled $X$.}
If one wishes to express $\CNOT$ using the manuscript's fully pattern-controlled primitives
$\CC^{(1)}_{v}(X)\in\mathrm U(2^j)$ (with $v\in\mathbb Z_2^{j-1}$ specifying the entire control word on $(q_2,\dots,q_j)$),
then for $r\in\{1,\dots,j-1\}$ one has
\begin{equation}\label{eq:CNOT_as_product_of_CC}
\CNOT(q_{r+1}\to q_1)
=
\prod_{u\in\mathbb{Z}_2^{j-2}}
\CC^{(1)}_{\,u_1\cdots u_{r-1}\,1\,u_r\cdots u_{j-2}}(X).
\end{equation}
The factors commute because they act on pairwise orthogonal control subspaces, and exactly one factor is active on any
branch with the $r$-th control bit equal to $1$.

\subsection{Stage-as-$\UCRY$ compilation and Gray-code ladder}\label{subsec:transpilation_compilation}

\paragraph{Rotation primitive.}
Recall that
\[
\rR(\alpha)=
\begin{pmatrix}
\cos\alpha & -\sin\alpha\\
\sin\alpha & \cos\alpha
\end{pmatrix}
=\RY(2\alpha),
\]
so at stage $j$ it suffices to compile controlled $\RY$ rotations with angles $\phi_{\mathbf{w}}:=2\theta_{\mathbf{w}}$.

\paragraph{Uniformly controlled viewpoint.}
Define the uniformly controlled $y$-rotation on the active register by
\[
\UCRY(\bm{\phi})
:=\prod_{\mathbf{w}\in\mathbb{Z}_2^{j-1}} \CC_{\mathbf{w}}^{(1)}\!\bigl(\RY(\phi_{\mathbf{w}})\bigr)\ \in\ \mathrm{U}(2^j),
\qquad
\bm{\phi}=\{\phi_{\mathbf{w}}\}_{\mathbf{w}\in\mathbb Z_2^{j-1}}.
\]
With $\phi_{\mathbf{w}}:=2\theta_{\mathbf{w}}$, the active stage unitary is precisely
\[
U_j^{\mathrm{(act)}}=\UCRY(\{2\theta_{\mathbf{w}}\}_{\mathbf{w}\in\mathbb Z_2^{j-1}}),
\qquad
\text{and the full-$n$ embedding is }\ \rI_{2^{n-j}}\otimes U_j^{\mathrm{(act)}}\in\mathrm U(2^n).
\]

\paragraph{Gray-code conventions.} Let $(\mathbf{g}_k)_{k=0}^{2^m-1}$ be the standard (binary-reflected) Gray code on $\mathbb{Z}_2^m$.
For $k\in\mathbb{Z}_{2^m}$ define
\[
\mathbf{g}_k:=b_m(k)\oplus b_m\!\left(\Big\lfloor\frac{k}{2}\Big\rfloor\right)\in\mathbb{Z}_2^m,
\qquad
\gamma(k):=b_m^{-1}(\mathbf{g}_k)\in\mathbb{Z}_{2^m},
\]
where $\oplus$ denotes componentwise XOR on $\mathbb{Z}_2^m$ and $b_m^{-1}$ is the inverse map
that sends a word $z_1\cdots z_m$ to the integer $\sum_{r=1}^m z_r2^{r-1}$.
Then consecutive words differ in exactly one bit.
For $k\ge 1$ set
\[
d_k:=\gamma(k)\oplus \gamma(k-1)\in\mathbb{Z}_{2^m},
\]
where $\oplus$ now denotes bitwise XOR on integers; in particular $d_k$ is a power of two.
Let $j(k)\in\{1,\dots,m\}$ be the unique index such that $d_k=2^{\,j(k)-1}$.
Because $b_m^{-1}$ uses the LSB-first encoding ($b_m^{-1}(z_1\cdots z_m)=\sum_r z_r 2^{r-1}$),
a flip in bit position $r$ of the word $\mathbf{g}_{k-1}\to\mathbf{g}_k$ contributes $\pm 2^{r-1}$ to
$\gamma(k)-\gamma(k-1)$, so $d_k=2^{r-1}$ and thus $j(k)=r$.
Equivalently, $j(k)$ is the unique bit position $r\in\{1,\dots,m\}$ that flips from $\mathbf{g}_{k-1}$ to $\mathbf{g}_k$,
where positions are counted from $z_1$ (LSB) to $z_m$ (MSB).

\begin{example}[Binary-reflected Gray code for $m=3$]\label{ex:graycode_m3}
We illustrate the binary-reflected Gray code using the paper's convention
\[
b_m:\mathbb{Z}_{2^m} \to \mathbb{Z}_2^m,\qquad
b_m(k)=z_1z_2\cdots z_m\ \Longleftrightarrow\ k=\sum_{r=1}^m z_r\,2^{r-1},
\]
(i.e.\ $z_1$ is the least significant bit). For $m=3$ (so $2^m=8$) and $k=0,\dots,7$, define the Gray word
\[
\mathbf{g}_k:= b_3(k)\oplus b_3\!\left(\Big\lfloor\frac{k}{2}\Big\rfloor\right)\in\mathbb{Z}_2^3,
\qquad
\gamma(k):=b_3^{-1}(\mathbf{g}_k)\in\{0,\dots,7\},
\]
where $\oplus$ denotes componentwise XOR on $\mathbb{Z}_2^3$ and $b_3^{-1}$ is the inverse map.

The table below lists $b_3(k)$, $b_3(\lfloor k/2\rfloor)$, $\mathbf{g}_k$,
the integer $\gamma(k)=b_3^{-1}(\mathbf{g}_k)$ computed with the LSB-first convention
$b_3^{-1}(z_1z_2z_3)=z_1+2z_2+4z_3$, and the flip mask $d_k=\gamma(k)\oplus\gamma(k-1)$
(integer XOR, undefined for $k=0$):

\medskip
\begin{center}
\begin{tabular}{c|c|c|c|c|c}
$k$ & $b_3(k)$ & $b_3(\lfloor k/2\rfloor)$ & $\mathbf{g}_k$ & $\gamma(k)$ & $d_k$\\
\hline
$0$ & $000$ & $000$ & $000$ & $0$ & $-$\\
$1$ & $100$ & $000$ & $100$ & $1$ & $1$\\
$2$ & $010$ & $100$ & $110$ & $3$ & $2$\\
$3$ & $110$ & $100$ & $010$ & $2$ & $1$\\
$4$ & $001$ & $010$ & $011$ & $6$ & $4$\\
$5$ & $101$ & $010$ & $111$ & $7$ & $1$\\
$6$ & $011$ & $110$ & $101$ & $5$ & $2$\\
$7$ & $111$ & $110$ & $001$ & $4$ & $1$
\end{tabular}
\end{center}
\medskip

Thus, in the paper's bit ordering, the Gray-code sequence is
\[
\mathbf{g}_0=000,\;
\mathbf{g}_1=100,\;
\mathbf{g}_2=110,\;
\mathbf{g}_3=010,\;
\mathbf{g}_4=011,\;
\mathbf{g}_5=111,\;
\mathbf{g}_6=101,\;
\mathbf{g}_7=001,
\]
and consecutive words differ in exactly one bit. For $k\ge 1$, define the (integer) flip mask
\[
d_k:=\gamma(k)\oplus \gamma(k-1)\in\{1,2,4\},
\]
where $\oplus$ now denotes bitwise XOR on integers.
Since $\gamma(k)$ is computed with the LSB-first encoding ($b_3^{-1}(z_1z_2z_3)=z_1+2z_2+4z_3$),
the bit $z_r$ contributes $2^{r-1}$ to $\gamma$, so a flip in position $r$ produces $d_k=2^{r-1}$.
In this example one obtains
\[
(d_1,\dots,d_7)=(1,2,1,4,1,2,1),
\]
so each $d_k$ is a power of two. The flip index $j(k)\in\{1,2,3\}$ is the unique integer such that
$d_k=2^{\,j(k)-1}$, hence
\[
(j(1),\dots,j(7))=(1,2,1,3,1,2,1).
\]
Equivalently, $j(k)$ records the unique bit position (in the word $z_1z_2z_3$) that flips from $\mathbf{g}_{k-1}$ to $\mathbf{g}_k$.
For instance,
\[
\mathbf{g}_1=100 \longrightarrow \mathbf{g}_2=110
\]
flips $z_2$ (the second bit), which corresponds to $d_2=\gamma(2)\oplus\gamma(1)=3\oplus 1=2=2^{2-1}$
and thus $j(2)=2$.
\end{example}

\begin{proposition}[Gray-code ladder for a uniformly controlled $\RY$ gate in $\mathrm{U}(2^j)$]
\label{prop:ucr_gray_ladder}
Let $j\ge 2$ and consider the active register $(q_1,\dots,q_j)$, where $q_1$ is the target and
$C=(q_2,\dots,q_j)$ are the $m:=j-1$ control qubits.
Given an angle list $\bm{\phi}=\{\phi_{\mathbf{w}}\}_{\mathbf{w}\in\mathbb{Z}_2^{m}}$, define the uniformly controlled rotation
\begin{equation}\label{eq:ucr_def}
\UCRY(\bm{\phi})
:=\prod_{\mathbf{w}\in\mathbb{Z}_2^{m}} \CC_{\mathbf{w}}^{(1)}\!\bigl(\RY(\phi_{\mathbf{w}})\bigr)\ \in\ \mathrm{U}(2^j),
\end{equation}
i.e.\ $\UCRY(\bm{\phi})$ applies $\RY(\phi_{\mathbf{w}})$ to $q_1$ conditioned on the control register being $\ket{\mathbf{w}}$.

Let $(\mathbf{g}_k)_{k=0}^{2^{m}-1}$ be the standard (binary-reflected) Gray code on $\mathbb{Z}_2^m$ and let
$j(k)\in\{1,\dots,m\}$ denote the unique bit that flips from $\mathbf{g}_{k-1}$ to $\mathbf{g}_k$.
Define a new angle list $\bm{\alpha}=\{\alpha_{\mathbf{v}}\}_{\mathbf{v}\in\mathbb{Z}_2^{m}}$ by the Walsh--Hadamard transform
\begin{equation}\label{eq:alpha_from_phi}
\alpha_{\mathbf{v}}
:=\frac{1}{2^{m}}\sum_{\mathbf{w}\in\mathbb{Z}_2^{m}}(-1)^{\langle \mathbf{v},\mathbf{w}\rangle}\,\phi_{\mathbf{w}},
\qquad \mathbf{v}\in\mathbb{Z}_2^{m},
\end{equation}
where $\langle \mathbf{v},\mathbf{w}\rangle:=\sum_{r=1}^{m} v_r w_r \pmod 2$.
Then $\UCRY(\bm{\phi})$ admits the ancilla-free Gray-code ladder decomposition (cf.~\cite{Bergholm2005})
\begin{equation}\label{eq:ucr_ladder}
\UCRY(\bm{\phi})
=
\RY\!\bigl(\alpha_{\mathbf{g}_0}\bigr)_{q_1}\;
\prod_{k=1}^{2^{m}-1}
\Bigl[\CNOT\!\bigl(q_{1+j(k)}\to q_1\bigr)\;\RY\!\bigl(\alpha_{\mathbf{g}_k}\bigr)_{q_1}\Bigr].
\end{equation}
In particular, the circuit uses exactly $2^{m}$ single-qubit $\RY$ gates and $2^{m}-1$ $\CNOT$ gates and requires no ancillas.
\end{proposition}

\begin{proof}
\noindent\textbf{Gray-walk viewpoint.}
Fix a control word $\mathbf{u}=u_1\cdots u_m\in\mathbb{Z}_2^m$ and restrict the Gray-code
ladder circuit~\eqref{eq:ucr_ladder} to the branch $\mathbb{C}^2\otimes\ket{\mathbf{u}}$, where
$\mathbb{C}^2$ refers to the target qubit $q_1$ and $\ket{\mathbf{u}}$ to the control register
$(q_2,\dots,q_j)$.
On this branch, each $\CNOT(q_{1+j(k)}\to q_1)$ applies $X$ to the target if and only if the
corresponding control bit equals~$1$, hence it acts as $X^{u_{j(k)}}$ on $q_1$.
Therefore the branch-restricted unitary reads
\[
U_{\mathbf{u}}
=
\RY\!\bigl(\alpha_{\mathbf{g}_0}\bigr)
\prod_{k=1}^{2^m-1}\Bigl(X^{u_{j(k)}}\,\RY\!\bigl(\alpha_{\mathbf{g}_k}\bigr)\Bigr),
\]
where we omit the target subscript $q_1$ for readability.

Let $S_k(\mathbf u):=\sum_{\ell=1}^{k} u_{j(\ell)}\pmod 2$ be the parity of the number of $X$ toggles
applied up to step $k$ on this branch. Using the elementary conjugation identity
\[
X\,\RY(\alpha)\,X=\RY(-\alpha)\qquad(\alpha\in\mathbb{R}),
\]
we can rewrite each factor $\RY(\alpha_{\mathbf{g}_k})$ as conjugated by the cumulative toggle $X^{S_{k-1}(\mathbf u)}$,
so that the effective angle at step $k$ acquires a sign
\[
\RY(\alpha_{\mathbf{g}_k})
\;\longmapsto\;
\RY\!\bigl((-1)^{S_{k-1}(\mathbf u)}\,\alpha_{\mathbf{g}_k}\bigr).
\]
For the binary-reflected Gray code, the parity $S_{k-1}(\mathbf{u})$ equals the $\mathbb{Z}_2$ inner product
$\langle \mathbf{u},\mathbf{g}_k\rangle:=\sum_{r=1}^m u_r (\mathbf{g}_k)_r\pmod 2$, hence the branch action simplifies to
\[
U_{\mathbf{u}}
=
\RY\!\Bigl(\sum_{k=0}^{2^m-1} (-1)^{\langle \mathbf{u},\mathbf{g}_k\rangle}\,\alpha_{\mathbf{g}_k}\Bigr)\,X^{S_{2^m-1}(\mathbf{u})}.
\]
Finally, along the Gray-code ladder each control index $r\in\{1,\dots,m\}$ appears exactly $2^{m-1}$ times
among the flips $j(1),\dots,j(2^m-1)$. Therefore
\[
S_{2^m-1}(\mathbf{u})=\sum_{r=1}^m u_r\,2^{m-1}\equiv 0\pmod 2
\qquad (m\ge 2),
\]
so the net trailing toggle is $X^{S_{2^m-1}(\mathbf{u})}=\rI$ and no residual Pauli factor remains.
Consequently, on each branch $\mathbf{u}$ the ladder realizes a \emph{pure} $\RY$ rotation with angle
$\sum_{k}(-1)^{\langle \mathbf{u},\mathbf{g}_k\rangle}\alpha_{\mathbf{g}_k}$, which is precisely the uniformly controlled action
once $\bm{\alpha}$ is chosen as the Walsh--Hadamard transform of $\bm{\phi}$ in \eqref{eq:alpha_from_phi}.

\smallskip
\noindent\textbf{Evaluation of the effective branch angle (Walsh--Hadamard inversion).}
From the Gray-walk analysis above (and the fact that the trailing toggle cancels for $m\ge2$), the restriction of
the ladder circuit \eqref{eq:ucr_ladder} to the control branch $\mathbb{C}^2 \otimes \ket{\mathbf{u}}$ is a pure rotation
\[
U_{\mathbf{u}}=\RY\!\bigl(\widetilde{\phi}(\mathbf{u})\bigr),
\qquad
\widetilde{\phi}(\mathbf{u}):=\sum_{k=0}^{2^m-1}(-1)^{\langle \mathbf{u},\mathbf{g}_k\rangle}\,\alpha_{\mathbf{g}_k},
\]
(for $m=1$ the ladder reduces to the standard two-gate controlled-$\RY$ synthesis and the same conclusion holds).
Since the Gray code $(\mathbf{g}_k)_{k=0}^{2^m-1}$ enumerates $\mathbb{Z}_2^m$ exactly once, we can re-index the sum by
$\mathbf{v}\in\mathbb{Z}_2^m$ and write
\[
\widetilde{\phi}(\mathbf{u})
=
\sum_{\mathbf{v}\in\mathbb{Z}_2^m}(-1)^{\langle \mathbf{u},\mathbf{v}\rangle}\,\alpha_{\mathbf{v}}.
\]
Using the definition \eqref{eq:alpha_from_phi},
\[
\alpha_{\mathbf{v}}
=\frac{1}{2^m}\sum_{\mathbf{w}\in\mathbb{Z}_2^m}(-1)^{\langle \mathbf{v},\mathbf{w}\rangle}\,\phi_{\mathbf{w}},
\]
we obtain, by exchanging sums,
\[
\widetilde{\phi}(\mathbf{u})
=\frac{1}{2^m}\sum_{\mathbf{w}\in\mathbb{Z}_2^m}\phi_{\mathbf{w}}
\sum_{\mathbf{v}\in\mathbb{Z}_2^m}(-1)^{\langle \mathbf{u}+\mathbf{w},\mathbf{v}\rangle}.
\]
By orthogonality of characters of the group $(\mathbb{Z}_2^m,+)$,
\[
\sum_{\mathbf{v}\in\mathbb{Z}_2^m}(-1)^{\langle \mathbf{a},\mathbf{v}\rangle}
=
\begin{cases}
2^m, & \mathbf{a}=\mathbf{0},\\
0, & \mathbf{a}\neq \mathbf{0},
\end{cases}
\]
hence only the term $\mathbf{w}=\mathbf{u}$ survives and $\widetilde{\phi}(\mathbf{u})=\phi_{\mathbf{u}}$.
Therefore $U_{\mathbf{u}}=\RY(\phi_{\mathbf{u}})$ on every branch $\mathbf{u}$, which is exactly the defining action of
$\UCRY(\bm{\phi})$ in \eqref{eq:ucr_def}.
\end{proof}

\paragraph{Conclusion: transpiling the Grover--Rudolph stage unitaries.}
At Grover--Rudolph stage $j\in\{2,\dots,n\}$ the circuit applies, on the active $j$-qubit register,
the product of pattern-controlled rotations
\[
U_j^{\mathrm{(act)}}=\prod_{\mathbf{w}\in\mathbb{Z}_2^{j-1}}\CC_{\mathbf{w}}^{(1)}\!\bigl(\rR(\theta_{\mathbf{w}})\bigr)
=\prod_{\mathbf{w}\in\mathbb{Z}_2^{j-1}}\CC_{\mathbf{w}}^{(1)}\!\bigl(\RY(2\theta_{\mathbf{w}})\bigr).
\]
Equivalently, $U_j^{\mathrm{(act)}}$ is a uniformly controlled $\RY$ gate with angle list
$\bm{\phi}=\{\phi_{\mathbf{w}}\}_{\mathbf{w}\in\mathbb{Z}_2^{j-1}}$ given by $\phi_{\mathbf{w}}:=2\theta_{\mathbf{w}}$.
Proposition~\ref{prop:ucr_gray_ladder} then yields an explicit ancilla-free Gray-code ladder realization of
$U_j^{\mathrm{(act)}}$ over the gate dictionary
$\mathcal{G}=\{X,\RY(\cdot),\CNOT(\cdot\to\cdot)\}$, using exactly $2^{j-1}$ $\RY$ gates and $2^{j-1}-1$ $\CNOT$ gates.
When embedded into the full $n$-qubit space, this acts as $\rI_{2^{n-j}}\otimes U_j^{\mathrm{(act)}}\in\mathrm{U}(2^n)$
(see Figure~\ref{fig:graycode_ladder_generic} for the ladder pattern).

\begin{remark}[practical compilation guideline]\label{rem:compile_stage}
In implementations, it is typically advantageous to transpile each Grover--Rudolph stage as a single
uniformly controlled $\RY$ gate (multiplexor) rather than transpiling the $2^{j-1}$ pattern-controlled factors
individually; this avoids redundant synthesis and directly yields the Gray-code ladder structure of
Proposition~\ref{prop:ucr_gray_ladder} (cf.~\cite{Bergholm2005}).
\end{remark}

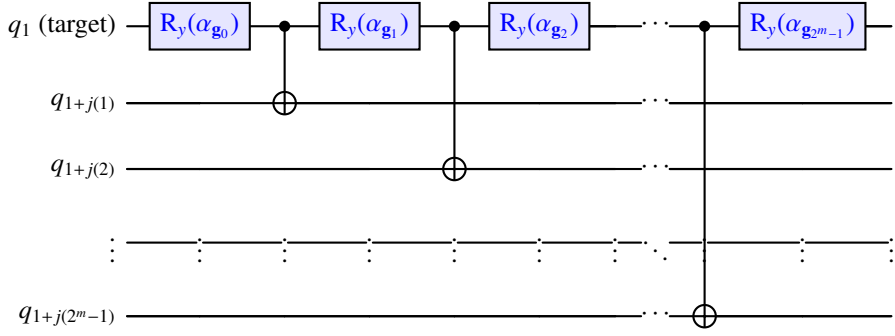
\begin{figure}[t]
\centering
\begin{quantikz}[row sep=0.55cm, column sep=0.30cm, wire types={q}]
\lstick{$q_1$ (target)} &
\gate[style={fill=blue!10},label style={blue}]{\RY(\alpha_{\mathbf{g}_0})} &
\ctrl{1} & \gate[style={fill=blue!10},label style={blue}]{\RY(\alpha_{\mathbf{g}_1})} &
\ctrl{2} & \gate[style={fill=blue!10},label style={blue}]{\RY(\alpha_{\mathbf{g}_2})} &
\qw & \cdots &
\ctrl{4} & \gate[style={fill=blue!10},label style={blue}]{\RY(\alpha_{\mathbf{g}_{2^{m}-1}})} &
\qw \\
\lstick{$q_{1+j(1)}$} &
\qw &
\targ{} & \qw &
\qw & \qw &
\qw & \cdots &
\qw & \qw &
\qw \\
\lstick{$q_{1+j(2)}$} &
\qw &
\qw & \qw &
\targ{} & \qw &
\qw & \cdots &
\qw & \qw &
\qw \\
\lstick{$\vdots$} &
\vdots &
\vdots & \vdots &
\vdots & \vdots &
\vdots & \ddots &
\vdots & \vdots &
\vdots \\
\lstick{$q_{1+j(2^{m}-1)}$} &
\qw &
\qw & \qw &
\qw & \qw &
\qw & \cdots &
\targ{} & \qw &
\qw \\
\end{quantikz}
\caption{Gray-code ladder realization of $\CC_{\mathbf 1}^{(1)}(\RY(\phi))$ on the active register $(q_1,\dots,q_j)$, where $m=j-1$.
At step $k$ the $\CNOT$ uses control $q_{1+j(k)}$ and target $q_1$, followed by $\RY(\alpha_{\mathbf{g}_k})$ on $q_1$.}
\label{fig:graycode_ladder_generic}
\end{figure}

\subsection*{Pseudo-code: ancilla-free Gray-code ladder compilation of a Grover--Rudolph stage as $\UCRY$}
\label{subsec:pseudocode_ucry}

\begin{algorithm}[H]
\caption{Ancilla-free Gray-code ladder compilation of $\UCRY(\bm{\phi})\in\mathrm{U}(2^j)$}
\label{alg:ucry_gray_ladder}
\KwIn{
Active register $(q_1,\dots,q_j)$ with target $q_1$ and controls $(q_2,\dots,q_j)$;\\
angle list $\bm{\phi}=\{\phi_{\mathbf{w}}\}_{\mathbf{w}\in\mathbb{Z}_2^{m}}$ with $m\leftarrow j-1$.
(In Grover--Rudolph stage $j$: $\phi_{\mathbf{w}} \leftarrow 2\theta_{\mathbf{w}}$.)
}
\KwOut{
A gate list over $\{\RY(\cdot),\CNOT(\cdot\to\cdot)\}$ implementing $\UCRY(\bm{\phi})$.
}

$m\leftarrow j-1$\;

\BlankLine
\textbf{(A) Compute Walsh--Hadamard angles $\bm{\alpha}$:}\;
\ForEach{$\mathbf{v}\in\mathbb{Z}_2^{m}$}{
    $\alpha_{\mathbf{v}} \leftarrow \dfrac{1}{2^{m}}\displaystyle\sum_{\mathbf{w}\in\mathbb{Z}_2^{m}} (-1)^{\langle \mathbf{v},\mathbf{w}\rangle}\,\phi_{\mathbf{w}}$\;
}
\tcp{In practice, compute $\bm{\alpha}$ with a fast Walsh--Hadamard transform in $O(m2^m)$.}

\BlankLine
\textbf{(B) Gray-code ladder order on the control space:}\;
\For{$k\leftarrow 0$ \KwTo $2^m-1$}{
    $\mathbf g_k \leftarrow b_m(k)\oplus b_m(\lfloor k/2\rfloor)\in\mathbb{Z}_2^m$\;
    $\gamma_k \leftarrow b_m^{-1}(\mathbf g_k)\in \mathbb{Z}_{2^m}$\;
    \If{$k=0$}{
        Apply $\RY(\alpha_{\mathbf g_0})$ on $q_1$\;
    }\Else{
        $d \leftarrow \gamma_k \oplus \gamma_{k-1}$ \tcp*{integer XOR; $d$ is a power of two}
        choose the unique $s\in\{1,\dots,m\}$ such that $d=2^{s-1}$\;
        Apply $\CNOT(q_{1+s}\to q_1)$\;
        Apply $\RY(\alpha_{\mathbf g_k})$ on $q_1$\;
    }
}

\end{algorithm}

\noindent\textbf{Note.}
For the Grover--Rudolph construction, set $\phi=2\theta_{\mathbf{w}}$, since $\rR(\theta_{\mathbf{w}})=\RY(2\theta_{\mathbf{w}})$.

\section*{Conclusions}

We have presented a rigorous and self-contained proof of correctness of the Grover--Rudolph state
preparation algorithm. The analysis makes explicit the dyadic refinement structure underlying the method,
derives a trigonometric factorization of the target probabilities in terms of conditional masses, and proves
by induction that the resulting hierarchy of controlled rotations produces exactly the desired measurement
law in the computational basis.

Beyond exact correctness, we established a quantitative stability theorem for imperfectly implemented
Grover--Rudolph angles (Theorem~\ref{thm:gr_angle_stability}).
The estimate shows that, if every Grover--Rudolph angle is perturbed by at most
\(\eta\), then the resulting output distribution differs from the target distribution by at most \(\min(1,n\eta)\)
in total variation distance. This linear-in-depth bound exploits the conditional binary-tree structure of
the construction. Combining this estimate with the Hoeffding concentration inequality,
Corollary~\ref{cor:combined_bound} gives a unified closed-form guarantee: for $\varepsilon\in(0,1]$, with probability at least
$1-\delta$,
\[
\mathrm{TV}(p,\hat p)\;\le\;\min\!\left(1,\;\frac{n\pi}{2^{b+1}}+\sqrt{\frac{2^n\log(2/\delta)}{2S}}\right),
\]
{ where $b$ is the angle bit depth and $S$ the number of measurement shots.
Inverting this bound yields an explicit design rule: $b\ge\log_2(2n\pi/\varepsilon)$ bits and
$S\ge 2^{n+1}\log(2/\delta)/\varepsilon^2$ shots suffice to achieve $\mathrm{TV}(p,\hat p)\le\varepsilon$
with confidence $1-\delta$. A key practical consequence (Remark~\ref{rem:separation}) is that the
required bit depth $b$ grows only \emph{logarithmically} in $n$ and $1/\varepsilon$, whereas the shot
budget $S$ must grow \emph{linearly} in $N=2^n$; the two error sources therefore have qualitatively
different scaling, and once $b$ is large enough, increasing it further provides no benefit.

The numerical experiments of Section~\ref{subsec:sensitivity} confirm these theoretical predictions
for the triangle density with $n\in\{2,3,4\}$ qubits: the angle-precision error at $b=8$ bits
is at most $4\times10^{-3}$ in TV, already dominated by shot noise at all tested shot counts
$S\in\{256,1024,4096\}$, consistently with the $O(\sqrt{2^n/S})$ rate of the Hoeffding bound.}

Finally, as a practical complement, we described an ancilla-free compilation viewpoint in which each
Grover--Rudolph stage is treated as a uniformly controlled one-qubit rotation, or multiplexor, admitting a
Gray-code ladder decomposition. This provides a direct route to gate-set transpilation on devices with a
native dictionary such as \(\{\RY(\cdot),X,\CNOT\}\), while remaining logically separate from the proof of the
Grover--Rudolph theorem.

Possible directions for future work include: (i) sharper \emph{distribution-dependent} stability
estimates that exploit cancellations in the dyadic tree or non-uniform angular error profiles,
beyond the worst-case uniform bound of Theorem~\ref{thm:gr_angle_stability};
(ii) hardware-aware cost models for compiling the stage multiplexors under connectivity constraints;
and (iii) extensions to structured families of distributions, such as log-concave or sparse
distributions, in the spirit of~\cite{Ramacciotti2024}, where additional synthesis savings
may be achievable.

\section*{Author contributions}
All authors contributed equally to this work.
Conceptualization, Methodology, Formal analysis, Investigation, Writing -- original draft,
Writing -- review \& editing: All authors.

\section*{Use of Generative-AI tools declaration}

The authors used generative-AI tools for language editing and/or formatting assistance. All scientific content,
derivations, and conclusions were produced and verified by the authors, who take full responsibility for the
manuscript.

\section*{Acknowledgments}

This work was supported by the Generalitat Valenciana under grant
COMCUANTICA/007 (QUANTWin), by the Agreement between the Directorate-General for Innovation of the
Ministry of Innovation, Industry, Trade and Tourism of the Generalitat Valenciana and the
Universidad CEU Cardenal Herrera, and by Universidad CEU Cardenal Herrera under grants INDI25/17 and GIR25/14.
The funders had no role in the study design; in the collection, analysis, or interpretation of data; in the writing
of the manuscript; or in the decision to publish the results.

\section*{Conflict of interest}
The authors declare that they have no conflict of interest.

\end{document}